\documentclass[reqno]{amsart}

\usepackage{amsmath,amsfonts}
\usepackage{epsfig}
\usepackage{graphicx}

\def\Lj{m_j+1}

\def\opDelta{\widehat{\Delta}}
\def\opB{\widehat{B}}

\def\opJ{J}
\def\opL{\widehat{L}}

\def\psidist{\widetilde{\psi}} 
\def\psidistex{\widetilde{\psi}}

\def\tr{{\rm Tr}}

\def\N{{\mathbb N}}

\def\R{{\mathbb R}}

\def\ut{{\underline{t}}}

\def\ux{{\underline{x}}}

\def\frH{{\mathfrak H}}
 
\def\cB{{\mathcal B}}

\def\cH{{\mathcal H}}

\def\cM{{\mathcal M}}
\def\cN{{\mathcal N}}

\def\1{{\bf 1}}

\def\eqnn{\begin{eqnarray*}}
\def\eeqnn{\end{eqnarray*}}
\def\eqn{\begin{eqnarray}}
\def\eeqn{\end{eqnarray}}

\def\prf{\begin{proof}}
\def\endprf{\end{proof}}

\theoremstyle{plain}
\newtheorem{theorem}{Theorem}[section]
\newtheorem{definition}[theorem]{Definition}
\newtheorem{proposition}[theorem]{Proposition}
\newtheorem{hypothesis}{Hypothesis}[section]

\newtheorem{lemma}[theorem]{Lemma}

\numberwithin{equation}{section}

\begin{document}

\parskip=8pt

\title[GP hierarchy and Quantum de Finetti]
{Unconditional uniqueness  for the cubic Gross-Pitaevskii hierarchy via quantum de Finetti}

\author[T. Chen]{Thomas Chen}
\address{T. Chen,  
Department of Mathematics, University of Texas at Austin.}
\email{tc@math.utexas.edu}

\author[C. Hainzl]{Christian Hainzl}
\address{C. Hainzl,
Fachbereich Mathematik,  Universit\"at T\"ubingen, Germany.}
\email{christian.hainzl@uni-tuebingen.de}

\author[N. Pavlovi\'{c}]{Nata\v{s}a Pavlovi\'{c}}
\address{N. Pavlovi\'{c},  
Department of Mathematics, University of Texas at Austin.}
\email{natasa@math.utexas.edu}

\author[R. Seiringer]{Robert Seiringer}
\address{R. Seiringer,
Institute of Science and Technology Austria (IST Austria).}
\email{robert.seiringer@ist.ac.at}

\maketitle

\begin{abstract}
We present a new, simpler proof of the  unconditional uniqueness of solutions to the cubic 
Gross-Pitaevskii hierarchy in $\R^3$. 
One of the main tools in our analysis is  the quantum de Finetti theorem.
Our uniqueness result is equivalent to the one established in the
celebrated works of Erd\"os, Schlein and Yau, \cite{esy1,esy2,esy3,esy4}.
\end{abstract}

\section{Introduction} 
\label{sec-mainres-1}

In this paper, we give a new proof of unconditional uniqueness of solutions to the cubic 
Gross-Pitaevskii (GP) hierarchy in $\R^3$.   
The cubic GP hierarchy is a system of infinitely many coupled linear PDE's
describing a Bose gas of infinitely many
particles, interacting via repulsive two-body delta interactions in the defocusing case, and 
via attractive two-body delta interactions in the focusing case.
In the defocusing case, it emerges from the $N\rightarrow\infty$ limit of the BBGKY hierarchy of marginal density matrices
for a bosonic $N$-particle Schr\"odinger system where the pair interaction potentials tend to a delta
distribution as $N\rightarrow\infty$.
Factorized solutions to  GP hierarchies are determined by solutions of the
corresponding cubic nonlinear Schr\"odinger (NLS) equation.
In this sense, the NLS is interpreted as the mean field description of an infinite system of interacting
bosons in the Gross-Pitaevskii limit.

The derivation of the nonlinear Hartree (NLH) equation from an interacting Bose gas
was first given by Hepp in \cite{he}. 
The BBGKY hierarchy was prominently used in the works of
Lanford for the study of classical mechanical systems in the infinite particle limit \cite{Lan-1,Lan-2}.
Subsequently, the first derivation of the NLH via the BBGKY hierarchy was given
by Spohn in \cite{sp}. More recently, this topic was revisited by Fr\"ohlich, Tsai and Yau in \cite{frtsya},
and in the last few years, Erd\"os, Schlein and Yau have
further developed the  BBGKY hierarchy approach to the derivation of 
the NLH and NLS in their landmark works \cite{esy1,esy2,esy3,esy4}, which initiated 
much of the current widespread interest in this research topic.
The proof strategy can be briefly summarized as follows.
We consider $N$ bosons in $\R^3$ 
described by the wave function 
$\Phi_{N}\in L_{sym}^2(\R^{3N})$, where $\Phi_{N}(x_1,\dots,x_N)$
is symmetric under permutation of particle variables,
and satisfies the  Schr\"odinger equation
\begin{equation}\label{ham1}
i\partial_{t}\Phi_{N} \, = \, H_{N}\Phi_{N} \,
\end{equation} 
with $N$-body Hamiltonian  
\begin{equation}\label{ham2}
	H_{N} \, = \, \sum_{j=1}^{N}(-\Delta_{x_{j}})+\frac1N\sum_{1\leq i<j \leq N}V_N(x_{i}-x_{j}) \,.
\end{equation}
The pair interaction potential has the form $V_N(x)=N^{3\beta}V(N^\beta x)$ with $\beta\in(0,1]$. 
$V$ is a sufficiently regular, rotationally symmetric pair interaction
potential (we will not specify $V$ in detail since the derivation of the NLS is not the topic of this paper). We assume that
$\int V(x)dx=1$ if $\beta<1$, and that the scattering length corresponding to $V$ has value 1 if $\beta=1$. 
Clearly, $V_N\rightharpoonup(\int V dx) \, \delta$ as $N\rightarrow\infty$.

For $k=1,\dots,N-1$, the
$k$-particle marginal density matrices 
are obtained from   
\eqn
    \gamma_{N}^{(k)}:=\tr_{k+1, k+2,...,N}|\Phi_{N}\rangle\langle\Phi_{N}| \,,
\eeqn
where $\tr_{k+1, k+2,...,N}$ denotes the partial trace with respect to the particle variables indexed
by $k+1, k+2,..., N$.  
It follows immediately that the property of {\em admissibility} holds,
\eqn\label{eq-BBGKYadm-1}
	\gamma^{(k)}_{N} \, = \, \tr_{k+1}(\gamma^{(k+1)}_{N})
	\; \; \; \; \;  , \; \; \; \; k\, = \, 1,\dots,N-1 \,,
\eeqn 
for $1\leq k\leq N-1$, and that 
$\tr (\gamma_{N}^{(k)})=\|\Phi_N\|_{L^{2} }^2=1$ for all $k<N$.

The $k$-particle marginals satisfy the $N$-particle BBGKY hierarchy 
\eqn\label{BBGKY}
	\lefteqn{
	i\partial_{t}\gamma_{N}^{(k)}(t,\ux_k;\ux_k')
	 \, = \, 
	-\sum_{j=1}^k(\Delta_{x_j}-\Delta_{x_j'})\gamma_{N}^{(k)}(t,\ux_k,\ux_k')
	}
	\nonumber\\
	&&
	+ \frac{1}{N}\sum_{1\leq i<j \leq k}[V_N(x_i-x_j)-V_N(x_i^{\prime}-x_{j}')]
	\gamma_{N}^{(k)}(t, \ux_{k};\ux_{k}') 
	\label{eq-bbgky-1}\\
	&&+\frac{N-k}{N}\sum_{i=1}^{k}\int dx_{k+1}[V_N(x_i-x_{k+1})-V_N(x_i^{\prime}-x_{k+1})]
	\nonumber\\  
	&&\hspace{6cm}
	\cdot\,\gamma_{N}^{(k+1)}(t, \ux_{k},x_{k+1};\ux_{k}',x_{k+1})\,.
	\nonumber
\eeqn  
In the limit as $N\rightarrow\infty$, solutions $\gamma_N^{(k)}$ to the BBGKY hierarchy 
tend to solutions $\gamma^{(k)}$ of the cubic, defocusing GP hierarchy, which we introduce below.
In \cite{esy1,esy2,esy3,esy4}, this is obtained from a weak-* limit in the trace class, for $\beta\in(0,1]$.
The case $\beta=1$ covered in  \cite{esy1,esy2,esy3,esy4} is a major achievement that is much harder than $\beta<1$. 
Strong convergence in a class of Hilbert-Schmidt type  
is obtained in \cite{CPBBGKY} for $\beta<\frac14$,
and subsequently in \cite{xch3,CheHol-2013} for $\beta<\frac23$.

The cubic GP hierarchy on $\R^3$  for an infinite sequence of 
bosonic
marginal density matrices $(\gamma^{(k)})_{k\in\N}$ is given by the infinite system of coupled linear PDE's
\begin{align} \label{eq-def-GP}
	i\partial_t \gamma^{(k)} &=\sum_{j=1}^k [-\Delta_{x_j},\gamma^{(k)}]   
	\, + \,   \lambda B_{k+1} \gamma^{(k+1)}   \,, \;\;k\in\N\,,
\end{align}
for suitable initial data $(\gamma^{(k)}(0))_{k\in\N}$,
where $\gamma^{(k)}(t;\ux_k;\ux_k')$ is fully symmetric under permutations separately of the
components of $\ux_k:=(x_1,\dots,x_k)$, and of the components of $\ux_k':=(x_1',\dots,x_k')$. 
The GP hierarchy is {\em defocusing} if $\lambda=1$, and {\em focusing} if $\lambda=-1$
(we are assuming the normalization condition $|\lambda|=1$ for simplicity).  
The interaction term for the $k$-particle marginal is defined by
\eqn \label{eq-def-b}
	B_{k+1}\gamma^{(k+1)}
	\, = \, B^+_{k+1}\gamma^{(k+1)}
        - B^-_{k+1}\gamma^{(k+1)} \, ,
\eeqn
where 
\eqn\label{eq-Bplus-GP-def-1}
    B^+_{k+1}\gamma^{(k+1)}
     = \sum_{j=1}^k B^+_{j;k+1 }\gamma^{(k+1)}
\eeqn
and 
\eqn 
    B^-_{k+1}\gamma^{(k+1)}
    = \sum_{j=1}^k B^-_{j;k+1 }\gamma^{(k+1)},
\eeqn                  
with 
\begin{align}\label{eq-Bplusmin-def-1-1}
    & \left(B^+_{j;k+1}\gamma^{(k+1)}\right)(t,x_1,\dots,x_k;x_1',\dots,x_k') \nonumber\\
    & \quad \quad = \int dx_{k+1}  dx_{k+1}'  \nonumber\\
    & \quad\quad\quad\quad 
	\delta(x_j-x_{k+1})\delta(x_j-x_{k+1}' )
        \gamma^{(k+1)}(t,x_1,\dots,x_{k+1};x_1',\dots,x_{k+1}') \nonumber\\
    & \quad \quad =  
        \gamma^{(k+1)}(t,x_1,\dots,x_j,\dots,x_k,x_j;x_1',\dots,x_k',x_j),
\end{align} 
and 
\begin{align}\label{eq-Bplusmin-def-1-2}
    & \left(B^-_{j;k+1}\gamma^{(k+1)}\right)(t,x_1,\dots,x_k;x_1',\dots,x_k') \nonumber\\
    & \quad \quad = \int dx_{k+1} dx_{k+1}'  \nonumber\\
    & \quad\quad\quad\quad 
	  \delta(x'_j-x_{k+1})\delta(x'_j-x_{k+1}' )
        \gamma^{(k+1)}(t,x_1,\dots,x_{k+1};x_1',\dots,x_{k+1}')\nonumber\\
    & \quad \quad =  
        \gamma^{(k+1)}(t,x_1,\dots,x_k,x_j';x_1',\dots,x_j',\dots,x_k',x_j') \,.
\end{align} 
We say that $B^+_{j;k+1}$ {\em contracts} the triple of variables $x_j,x_{k+1},x_{k+1}'$, and that
$B^-_{j;k+1}$ contracts the triple of variables $x_j',x_{k+1},x_{k+1}'$.

For $\alpha\geq0$, we define the spaces
\begin{align}\label{frakH-def-1}
 	\frH^\alpha:=\Big\{ \,(\gamma^{(k)})_{k\in\N} \, \Big| \, 
 	\tr(| S^{(k,\alpha)}  [\gamma^{(k)}] |) < M^{2k} \; \mbox{for some constant }M<\infty \, \Big\}
\end{align}
where  
\eqn\label{eq-Skalpha-def-1}
    S^{(k,\alpha)}[\gamma^{(k)}](\ux_k;\ux_{k}'):=\prod_{j=1}^k(1-\Delta_{x_j})^{\alpha/2}(1-\Delta_{x_j'})^{\alpha/2}
    \gamma^{(k)}(\ux_k;\ux_{k}')\,.
\eeqn

A {\em mild solution} in the space $L^\infty_{t\in[0,T)}\frH^1$, 
to the GP hierarchy with initial data $(\gamma^{(k)}(0))_{k\in\N} \in\frH^1$,  
is a solution of the integral equation
\eqn\label{eq-Duhamel-1}
    \gamma^{(k)}(t) = U^{(k)}(t)\gamma^{(k)}(0) + i \lambda \int_0^t  U^{(k)}(t-s) B_{k+1}\gamma^{(k+1)}(s) ds \,
    \;\;\;,\;\;\; k\in\N\,,
\eeqn
satisfying  
\eqn
    \sup_{t\in[0,T)}\tr(|S^{(k,1)} [\gamma^{(k)}(t)] |) < M^{2k}
\eeqn
for a finite constant $M$ independent of $k$.  
Here,  
\eqn
     U^{(k)}(t) := \prod_{\ell=1}^k e^{it(\Delta_{x_\ell}-\Delta_{x_\ell'})}
\eeqn
denotes the free $k$-particle propagator.
We note that 
\eqn\label{eq-UBcommut-1}
    e^{it(\Delta_{x_\ell}-\Delta_{x_\ell'})}B^\pm_{j;k+1}  
    = B^\pm_{j;k+1 } e^{it(\Delta_{x_\ell}-\Delta_{x_\ell'})} \;\;\;,\;\;\;\ell\not\in\{ j,k+1\} \,.
\eeqn
That is, any free propagator $e^{it(\Delta_{x_\ell}-\Delta_{x_\ell'})}$ commutes with $B^\pm_{j;k+1 }$ if
the variables $x_\ell$, $x_{\ell}'$ are not affected by  $B^\pm_{j;k+1 }$.

We remark that given factorized initial data,
\eqn
	\gamma^{(k)}_0(\ux_k;\ux_k') \, = \, \prod_{j=1}^k \phi_0(x_j) \, \overline{\phi_0(x_j')} \,,
\eeqn
the condition that $(\gamma^{(k)}(0))\in\frH^1$ is equivalent to 
\eqn\label{eq-trSkgamma-cond-1}
    \tr(|S^{(k,1)} [\gamma^{(k)}(0)] |) = \|\phi_0\|_{H^1}^{2k} < M^{2k} 
    \;\;\;,\;\;\;k\in\N\,,
\eeqn
that is, $\|\phi_0\|_{H^1} < M$ for some $M<\infty$.
Then, a solution to the GP hierarchy in $L^\infty_{t\in[0,T)}\frH^1$ having these initial data is given by
the sequence of {\em factorized} density matrices
\eqn
	\gamma^{(k)}(t;\ux_k;\ux_k') \,  =  \,
	\prod_{j=1}^{k} \phi_t(x_j) \, \overline{\phi_t(x'_j) }\,,
\eeqn
if the corresponding 1-particle wave function satisfies the cubic NLS 
\eqn\label{eq-NLS-def-1}
    i\partial_t\phi_t=-\Delta\phi_t+ \lambda|\phi_t|^2\phi_t \;\;\;,\;\;\;\phi_0\in H^1\,. 
\eeqn 
In this sense, the NLS is interpreted as the mean field description of an infinite system of interacting
bosons.
The  Cauchy problem \eqref{eq-NLS-def-1} is globally well-posed in the defocusing case $\lambda=1$,
and locally well-posed in the focusing case $\lambda=-1$, \cite{TaoBook}.
Solutions to \eqref{eq-NLS-def-1} conserve the $L^2$-mass $\|\phi_t\|_{L^2_x}=\|\phi_0\|_{L^2_x}$ and the energy,
\eqn
    E[\phi_t] = \frac12\| \nabla_x \phi_t \|_{L^2_x}^2 + \frac\lambda4\|\phi_t\|_{L^4_x}^4 = E[\phi_0] \,.
\eeqn
In particular,  $(\gamma^{(k)})_{k\in\N}\in L^\infty_{t\in[0,T)}\frH^1$ is equivalent to 
$\|\phi_t\|_{L^\infty_{t\in[0,T)}H^1}<M'$ for
some finite constant $M'$.

The {\em uniqueness of solutions} to the GP hierarchy in $L^\infty_{t\in[0,T)}\frH^1$ was 
established by Erd\"os, Schlein, and Yau in \cite{esy1,esy2,esy3,esy4}. 
This is a crucial, and very involved part of their program to derive
the cubic defocusing NLS as the mean field 
description of a bosonic $N$-body Schr\"odinger evolution, as $N\rightarrow\infty$.
Their uniqueness proof uses a sophisticated and extensive construction involving Feynman graph expansions,
and high dimensional singular integral estimates.
A key ingredient in their proof is a powerful combinatorial method  that resolves the problem of the
factorial growth of number of terms in iterated Duhamel expansions; we outline it in
Section \ref{ssec-KMsurvey-1}.

Subsequently,   Klainerman and Machedon \cite{KM} gave a much shorter proof of the uniqueness of solutions
to the GP hierarchy satisfying
\eqn
     \big\| \, \big[\prod_{j=1}^k|\nabla_{x_j}|\,|\nabla_{x_j'}| \big]\,  \gamma^{(k)} \, \big\|_{\rm HS} < c^k \;\;\;,\;\;\; k\in\N\,,
\eeqn
for  the Hilbert-Schmidt  norms
\eqn\label{eq-gamma-norm-def-1}
	\| \gamma^{(k)} \|_{\rm HS} &:=& \Big(\, \tr(|  \gamma^{(k)} |^2)\,\Big)^{1/2} 
	\nonumber\\
	&=&\Big( \, \int_{\R^{3k}\times\R^{3k}} |\gamma^{(k)}(\ux_k;\ux_k')|^2 d\ux_k d\ux_k' \, \Big)^{1/2} \,,
\eeqn
but conditional on the a priori assumption that
\eqn\label{eq-KMcondition-1}
    \Big\| \, \big\| \, \big[\prod_{j=1}^k|\nabla_{x_j}|\,|\nabla_{x_j'}| \big]\, B_{k+1}^\pm\gamma^{(k+1)} \, \big\|_{\rm HS} \,
    \, \Big\|_{L^1_{t\in[0,T)} } < C^k 
\eeqn
holds for some finite constants $c$, $C$ independent of $k$. 
We will refer to \eqref{eq-KMcondition-1} as the {\em Klainerman-Machedon condition}. 
Their approach uses techniques from the analysis of dispersive nonlinear PDEs,
together with the combinatorial method of Erd\"os, Schlein and Yau \cite{esy1,esy2,esy3,esy4}, 
which Klainerman and Machedon presented as the ``boardgame argument".
Starting with the work \cite{kiscst} for the cubic GP hierarchy on $\R^2$ and ${\mathbb T}^2$, 
the approach of Klainerman and Machedon was used by various authors for the derivation of the NLS from interacting Bose gases 
\cite{chpa,CPBBGKY,xch3,CheHol-2013,kiscst,CT1,zxie}.
The method of Klainerman and Machedon also 
inspired the analysis of the Cauchy problem for the GP hierarchy which was initiated in \cite{chpa2} and continued e.g. in \cite{GreSohSta-2012,CT1}.

The derivation of nonlinear dispersive PDEs,
such as the nonlinear Schr\"odinger (NLS) or nonlinear Hartree (NLH) equations, 
from many body quantum dynamics is a very active research topic,
and has been approached by many authors in a variety of ways;
see  \cite{esy1,esy2,esy3,esy4,ey,kiscst,rosc} and the references therein,
and also \cite{adgote,AmmariNier-2008,
AmmariNier-2011,anasig,frgrsc,frknpi,frknsc,grma,grmama,he,pick,pick2}.

This problem is closely related to the phenomenon of Bose-Einstein condensation (BEC)
in systems of interacting bosons, which was first experimentally verified in 1995,
\cite{anenmawico,dameandrdukuke}. For the mathematical study of BEC, we refer to 
\cite{ailisesoyn,lise,lisesoyn,liseyn, liseyn2} and the references therein. 


\section{Statement of main results}

The only currently available proof of unconditional uniqueness of solutions  in $L^\infty_{t\in[0,T)}\frH^1$
to the cubic GP hierarchy in $\R^3$ is 
given in the celebrated 
works of Erd\"os, Schlein, and Yau, \cite{esy1,esy2,esy3,esy4}, using an involved construction based on 
Feynman graph expansions and high-dimensional singular integral estimates.
The purpose of the paper at hand is to present a new, simpler proof. We note that  
this paper contains 
several extensive example calculations and detailed explanations of 
background material for the benefit of the reader,  
but the actual core of our proof, given in  in Sections 
\ref{sec-RecBounds-1} and \ref{sec-prooffinish-1}, is short.
We expect that our methods can be extended to solutions in
$L^\infty_{t\in[0,T)}\frH^s$ for some values $s<1$,  and to GP hierarchies in $\R^n$ with $n$ other
than 3.

\subsection{Prerequisites}
A key tool in our proof is the {\em quantum de Finetti} theorem, 
which is a quantum analogue of the Hewitt-Savage theorem in probability theory, \cite{HewittSavage}.
The strong version is due to Hudson-Moody, and Stormer, \cite{HudsonMoody,Stormer-69}, and
applies to  sequences of density matrices that are {\em admissible}, i.e.,  
\eqn
    \gamma^{(k)}=\tr_{k+1}(\gamma^{(k+1)})  \;\;\;\forall k\in\N \,,
\eeqn
similarly to \eqref{eq-BBGKYadm-1}.
We quote it in the formulation presented by Lewin, Nam and Rougerie in \cite{lnr} 
(\cite{HudsonMoody,Stormer-69} state it in the $C^*$-algebraic context).

\begin{theorem}\label{thm-strongDeFinetti-1}
(Strong Quantum de Finetti theorem, \cite{HudsonMoody,Stormer-69,lnr})
Let $\cH$ be any separable Hilbert space and let 
$\cH^k = \bigotimes_{sym}^k\cH$ denote the corresponding bosonic $k$-particle space. 
Let $\Gamma$ denote a collection of admissible 
bosonic density matrices on  $\cH$, i.e.,
$$
\Gamma = (\gamma^{(1)},\gamma^{(2)},\dots)
$$
with $\gamma^{(k)}$ a non-negative trace class operator on $\cH^k$, 
and $\gamma^{(k)}=\tr_{k+1} \gamma^{(k+1)}$, 
where $\tr_{k+1}$ denotes the partial trace over the $(k+1)$-th factor. 
Then, there exists a unique Borel probability measure $\mu$, 
supported on the unit sphere $S\subset\cH$, and invariant under multiplication of 
$\phi \in \cH$ by complex numbers of modulus one, such that 
\begin{equation}\label{gkdf}
    \gamma^{(k)} = \int d\mu(\phi)  (  | \phi  \rangle \langle \phi |  )^{\otimes k}
    \;\;\;,\;\;\;\forall k\in\N\,.
\end{equation} 
\end{theorem}

The limiting hierarchies of marginal density matrices obtained, for each value of the
time variable $t$, via weak-* limits from
the BBGKY hierarchy of bosonic $N$-body 
Schr\"odinger systems as in \cite{esy1,esy2,esy3,esy4} do not necessarily satisfy admissibility.
A weak version of the quantum de Finetti theorem then still applies;  in the form quoted in 
Theorem \ref{thm-weakDeFinetti-1}, below, it was 
recently proven by Lewin, Nam and Rougerie \cite{lnr}. 
Previously, Ammari and Nier proved an equivalent result in
\cite{AmmariNier-2008,AmmariNier-2011} in the context of $N$-body boson 
systems as $N\rightarrow\infty$, with
less singular interactions than those considered for the GP hierarchy.

\begin{theorem}\label{thm-weakDeFinetti-1}
(Weak Quantum de Finetti theorem, \cite{lnr,AmmariNier-2008,AmmariNier-2011}) 
Let $\cH$ be any separable Hilbert space and let 
$\cH^k = \bigotimes_{sym}^k\cH$ denote the corresponding bosonic $k$-particle space. 
Assume that $\gamma_N^{(N)}$ is an arbitrary sequence of mixed states on $\cH^N$, $N\in\N$,
satisfying $\gamma_N^{(N)}\geq 0$ and $\tr_{\cH^N}(\gamma_N^{(N)})=1$, and assume
that its $k$-particle marginals have weak-* limits 
\eqn 
    \gamma^{(k)}_{N}:=\tr_{k+1,\cdots,N}(\gamma^{(N)}_N)
    \; \rightharpoonup^* \; \gamma^{(k)} \;\;\;\; (N\rightarrow\infty)\,,
\eeqn
in the trace class on $\cH^k$ for all $k\geq1$ (here, $\tr_{k+1,\cdots,N}(\gamma^{(N)}_N)$ 
denotes the partial trace in the $(k+1)$-st up to $N$-th component). 
Then, there exists a unique Borel probability measure $\mu$ on the unit ball 
$\cB\subset\cH$, and invariant under multiplication of 
$\phi \in \cH$ by complex numbers of modulus one, 
such that
\eqref{gkdf} holds
for all $k\geq0$. 
\end{theorem}

In our case, we consider the Hilbert space $\cH=L^2(\R^3)$. 
For a detailed discussion of the strong and weak quantum de Finetti theorem, we refer to \cite{lnr},
where the notions of strong and weak quantum de Finetti were introduced.

\subsection{Main results}
Our main result is a new proof of the {\em unconditional}  uniqueness of solutions to the GP hierarchy
in $L^\infty_{t\in[0,T)}\frH^1$. 
We note that the property $(\gamma^{(k)})_{k\in\N}\in L^\infty_{t\in[0,T)}\frH^1$ implies that 
$\gamma^{(k)}(t)=\int d\mu_t(\phi)(|\phi\rangle\langle\phi|)^{\otimes k}$, 
where the measure $\mu_t$ has bounded
support in $H^1(\R^3)$; this is explained in Lemma \ref{lm-Chebyshev-1} below.

\begin{theorem}\label{thm-main-1}
Let $(\gamma^{(k)}(t))_{k\in\N}$ be a mild solution  in $L^\infty_{t\in[0,T)}\frH^1$ 
to the (de)focusing cubic GP hierarchy in $\R^3$ with initial data $(\gamma^{(k)}(0))_{k\in\N}\in\frH^1$, 
which is either admissible, or obtained
at each $t$ from a weak-* limit as described in  Theorem \ref{thm-weakDeFinetti-1}.  

Then, $(\gamma^{(k)})_{k\in\N}$ is the unique solution for the given initial data.

Moreover, assume that the initial data $(\gamma^{(k)}(0))_{k\in\N} \in\frH^1$ satisfy
\eqn 
    \gamma^{(k)}(0) = \int d\mu(\phi)(|\phi\rangle\langle\phi|)^{\otimes k} 
    \;\;\;,\;\;\;\forall k\in\N\,,
\eeqn  
(as guaranteed by Theorems \ref{thm-strongDeFinetti-1} and \ref{thm-weakDeFinetti-1})
where $\mu$ is a Borel probability measure
supported  either  on the unit sphere or on the unit ball in $L^2(\R^3)$,
and invariant under multiplication of 
$\phi \in \cH$ by complex numbers of modulus one.  
Then, 
\eqn 
    \gamma^{(k)}(t) = \int d\mu(\phi)(|S_t(\phi)\rangle\langle S_t(\phi)|)^{\otimes k} 
    \;\;\;,\;\;\;\forall k\in\N\,,
\eeqn
where $S_t:\phi\mapsto \phi_t$ is the flow map of the cubic (de)focusing NLS, for $t\in[0,T)$. That is, $\phi_t$
satisfies \eqref{eq-NLS-def-1} with initial data $\phi$.
\end{theorem}

Our proof of uniqueness uses the  Erd\"os-Schlein-Yau combinatorial method \cite{esy1,esy2,esy3,esy4},
in boardgame form as presented by Klainerman-Machedon in \cite{KM}.
However, we do {\em not} use the Klainerman-Machedon
condition \eqref{eq-KMcondition-1}, but instead apply the quantum de Finetti theorems. 
The uniqueness established in Theorem \ref{thm-main-1}  is unconditional, and
equivalent to the one proven in the
celebrated works of Erd\"os, Schlein and Yau, \cite{esy1,esy2,esy3,esy4}.

\subsubsection{Uniqueness of strong solutions}

Our next result addresses the uniqueness of  strong solutions, and shows the 
strength of de Finetti for the GP hierarchy.
We consider   
strong solutions to the GP hierarchy
$\Gamma\in C^1([0,T),\frH^{1})$
with $\partial_t\Gamma\in C([0,T),\frH^{-1})$  in $\R^d$ where $d=1,2,3$
(see \cite{caz} for a discussion on the level of the NLS). 
The quantum de Finetti theorem can be used for a short direct proof of uniqueness in this case, 
along the lines of Spohn's argument for the Vlasov hierarchy
\cite{Spohn81}.

\begin{theorem}\label{thm-main-3}
Let $d\in\{1,2,3\}$, and $\Gamma=(\gamma^{(k)})$.
Let 
$\opDelta\Gamma := (\sum_{j=1}^k [\Delta_{x_j},\gamma^{(k)}]  )_{k\in\N}$
and $\opB\Gamma:=(B_{k+1}\gamma^{(k+1)})$.
Let $\Gamma(t)=(\gamma^{(k)}(t))$ be a strong solution of the GP hierarchy 
in $C^1([0,T),\frH^{1})$ on $\R^d$, which is either admissible, or obtained
at each $t$ from a weak-* limit as described in  Theorem \ref{thm-weakDeFinetti-1}. 
Then, 
\begin{equation}\label{bbgky}
    i\partial_t \Gamma(t) = \opL \Gamma(t) \; \in\;C([0,T),\frH^{-1})
    \;\;\;,\;\;\; \opL = -\opDelta+\lambda\opB
    \;\;\;,\;\;\;\lambda\in\{1,-1\}\,,
\end{equation}
and $\Gamma(t)$ is the unique solution for the given initial data.
\end{theorem}

\section{Proof of Theorem \ref{thm-main-3}}
Let $\Gamma_\phi:=((|\phi\rangle\langle\phi|)^{\otimes k})_{k\in\N}$ for brevity.
We know that
$\int d\nu(\phi) \Gamma_{S_t(\phi)}$ solves (\ref{bbgky}), with $\Gamma_0 = \int d \nu(\phi) \Gamma_\phi$,  
if $S_t(\phi)$ is the flow map of the cubic (de)focusing NLS
\begin{equation}\label{nls}
i\partial_t \phi_t = -\Delta \phi_t +  \lambda|\phi_t|^2 \phi_t
\end{equation}
with initial data $\phi_0\in H^1(\R^d)$.

Since by assumption, $\Gamma(t)$  is either admissible, or obtained
at each $t$ from a weak-* limit as described in  Theorem \ref{thm-weakDeFinetti-1},
the de Finetti theorem implies that for every $t\geq 0$ there exists a Borel probability measure $\mu_t$ on $L^2(\R^d)$ such that
\begin{equation}\label{gkdf2}
    \Gamma(t) = \int d\mu_t(\phi)  \Gamma_\phi \,.
\end{equation}
Furthermore, we note that  for any $s\geq0$ and an arbitrary small $\eta>0$,
\eqn\label{eq-dernormbds-aux-1}
    \tr\big(\,\big| \, S^{(k+1,s-\frac d2-\eta)}[(\opB\Gamma)^{(k)}]\, \big|\,\big) 
    \leq C  \tr\big(\,\big| \, S^{(k+1,s)}[\Gamma^{(k+1)}]\, \big|\,\big)  \,.
\eeqn
This inequality was proven in \cite{CPHE} (it corresponds to a generalized Sobolev inequality for density matrices).
Moreover, it is evident that for any $s\geq0$,
\eqn\label{eq-dernormbds-aux-2}
    \tr\big(\,\big| \, S^{(k,s-2)}[(\opDelta\Gamma)^{(k)}]\, \big|\,\big) 
    \leq   \tr\big(\,\big| \, S^{(k,s)}[\Gamma^{(k)}]\, \big|\,\big)  
    \;\;\;,\;k\in\N\,.
\eeqn
Since $\frac d2+\eta\leq2$ for $d\leq3$ in \eqref{eq-dernormbds-aux-1}, we find that 
by setting $s=1$ in \eqref{eq-dernormbds-aux-1} and \eqref{eq-dernormbds-aux-2},
\eqn 
    \tr\big(\,\big| \, S^{(k,-1)}[ (\opL\Gamma)^{(k)} ]\, \big|\,\big) < C  
    \tr\big(\,\big| \, S^{(k,1)}[\Gamma^{(k)}]\, \big|\,\big) < C^k \;\;\;,\;k\in\N\,,
\eeqn
or in other words, $\opL\Gamma\in C([0,T),\frH^{-1})$. 
Here,   the second inequality follows from $\Gamma\in C^1([0,T),\frH^{1})$, see 
\eqref{eq-trSkgamma-cond-1}.
Since $\Gamma(t)$ is a strong solution of \eqref{bbgky} in $C^1([0,T),\frH^{1})$, 
we obtain that 
\eqn\label{eq-tildeGamma-aux-1-1}
    i\partial_t\Gamma^{(k)}(t) = 
    \lim_{h\rightarrow0}\frac1h \Big(\,\Gamma^{(k)}(t+h)-\Gamma^{(k)}(t)\,\Big)
    =(\opL \Gamma)^{(k)}(t) \,,
\eeqn  
hence $\partial_t\Gamma\in C([0,T),\frH^{-1})$. 
Then,  indeed,
\begin{equation}\label{diffmu}
    \frac d{dt} \int d\mu_t(\phi)  \Gamma_\phi =
    \lim_{h\rightarrow0}\frac1h\Big(\int  d\mu_{t+h}(\phi)\Gamma_\phi 
    -\int  d\mu_t(\phi)\Gamma_\phi \Big)=-i \int \opL \Gamma_\phi d\mu_t(\phi) 
\end{equation}
holds in  $C([0,T),\frH^{-1})$.
In this sense, $\mu_t$ is differentiable, with derivative given by the operator $\opL$.

In analogy to Spohn's argument in \cite{Spohn81}, we can show that the measure $\mu_t$ 
induces a flow on the unit ball which satisfies the GP-equation \eqref{nls}. 
To this end, we define 
\eqn
    \tilde \Gamma(t) := \int d\mu_t(\phi)  \Gamma_{S_{-t}(\phi)}.
\eeqn
Differentiating this with respect to $t$,  gives
\begin{align} 
    \frac d{dt}  \tilde \Gamma(t) &= \lim_{h\rightarrow0}\frac1h
    \Big(\, \int  d\mu_{t+h}(\phi)\Gamma_{S_{-t-h}(\phi)}-   \int  d\mu_{t}(\phi)\Gamma_{S_{-t}(\phi)}\, \Big)
    \nonumber\\
    &=
     \lim_{h\rightarrow0}\frac1h
    \Big(\, \int  d\mu_{t+h}(\phi)\Gamma_{S_{-t}(\phi)}-   \int  d\mu_{t}(\phi)\Gamma_{S_{-t}(\phi)}\, \Big)
     \label{eq-tildeGamma-aux-1-0}\\
    &\hspace{2cm}+ \lim_{h\rightarrow0}\frac1h
     \int  d\mu_{t+h}(\phi)\Big(\Gamma_{S_{-t-h}(\phi)}-  \Gamma_{S_{-t}(\phi)}\, \Big)
    \label{eq-tildeGamma-aux-1}\\
    &= 
     -i\int  d\mu_t  (\phi) \opL\Gamma_{S_{-t}(\phi)}
      + i\int  d\mu_t  (\phi) \opL \Gamma_{S_{-t}(\phi)} =0, \label{equ}
\end{align}
where we applied \eqref{diffmu} to \eqref{eq-tildeGamma-aux-1-0}
to get the first term in \eqref{equ}, and where we applied \eqref{eq-tildeGamma-aux-1-1} to
\eqref{eq-tildeGamma-aux-1} to get the second term in \eqref{equ}.   

Since the map $\phi \mapsto S_t(\phi)$ is a bijection on the intersection of the unit ball
of $L^2(\R^d)$ with $H^1(\R^d)$,
we obtain by a simple variable transformation
\eqn
     \int d\mu_t(\phi)  \Gamma_{S_{-t}(\phi)} = \int d\mu_t\big(S_t(\phi)\big)  \Gamma_\phi.
\eeqn
Because of \eqref{equ}, we find that
\eqn
    \tilde \Gamma(t) = \int d\mu_t\big(S_t(\phi)\big)  \Gamma_\phi
\eeqn
does not depend on $t$, and by the uniqueness parts of Theorems \ref{thm-strongDeFinetti-1}
and \ref{thm-weakDeFinetti-1}, we infer that   
 $$d\mu_t\big(S_t(\phi)\big)  = d\mu_0(\phi).$$
Hence, by variable transformation,
$$\Gamma(t) = \int d\mu_t(\phi)  \Gamma_\phi = \int d\mu_t\big(S_t(\phi)\big) \Gamma_{S_t(\phi)}
        = \int d\mu_0(\phi) \Gamma_{S_t(\phi)} .$$
By the uniqueness parts of Theorems \ref{thm-strongDeFinetti-1} and \ref{thm-weakDeFinetti-1}, 
$d\nu = d\mu_0$.  
\qed

\section{Proof of Theorem \ref{thm-main-1}}
   
The remainder of this paper is dedicated to the proof of Theorem \ref{thm-main-1}. 
As a preparation, we present three auxiliary combinatorial tools used for the organization of 
Duhamel expansion terms, in Sections  \ref{ssec-KMsurvey-1}, \ref{sec-treegraph-1}, and \ref{sec-disttree-1}.
For the convenience  of the reader, 
we will give a detailed survey of background material in these parts, and present several detailed example calculations. 
The core of our proof is contained in Sections 
\ref{sec-RecBounds-1} and \ref{sec-prooffinish-1}, and is short. 

\subsection{The Erd\"os-Schlein-Yau combinatorial method in  boardgame form}
\label{ssec-KMsurvey-1}
To prove   uniqueness, we will show that the trace norm of $\gamma^{(k)}$ (instead of
the Hilbert-Schmidt norm of $S^{(k,1)}[\gamma^{(k)}]$
as in \cite{KM}) is zero if the initial data is zero (see Section \ref{ssec-setupproof-1} for details). 
We will use the powerful combinatorial method of Erd\"os, Schlein and Yau
\cite{esy1,esy2,esy3,esy4}, which was presented in an elegant and accessible form by 
Klainerman and Machedon in \cite{KM} as a "boardgame argument".

To begin with, we consider the $r$-fold iterate of the Duhamel formula
\eqref{eq-Duhamel-1} for   $\gamma^{(k)}$, with initial data $\gamma^{(k)}_0=0$, 
for some arbitrary $r\in\N$,
\begin{align}
    \gamma^{(k)}(t ) \,&=\,(i\lambda)^r
    \int_{t\geq t_1\geq \cdots\geq t_r} dt_{1}\cdots dt_r  U^{(k)}(t-t_1)
    B_{ k+1}U^{(k+1)}(t_1-t_2)\cdots  
    \nonumber\\
    & \hspace{1cm}
    \cdots
    U^{(k+r-1)}(t_{r-1}-t_{r})B_{ k+r}
    \gamma^{ (k+r)}(t_r) 
    \nonumber\\
    &=:\,
    \int_{t\geq t_1\geq\cdots\geq t_r} dt_{1}\cdots dt_r J^k(\ut_r) 
    \;\;\;,\;\;\; \ut_r:=(t_1,\dots,t_r) \,.
    \label{eq-Duham-r-1-1} 
\end{align}
We will prove the following main lemma.

\begin{lemma}\label{lm-mainlm-trest-1}
Assume that $(\gamma^{(k)}(t))$ is a mild solution to the cubic GP hierarchy \eqref{eq-def-GP}
with initial data $\gamma^{(k)}(0)=0$ for all $k$, which is either admissible or obtained at each $t$
from a weak-* limit as in  Theorem \ref{thm-weakDeFinetti-1}.
Moreover, assume that 
\eqn\label{eq-Skgamma-cond-0}
    \sup_{t\in[0,T)}\tr(|S^{(k,1)}[\gamma^{(k)}(t)] |) < M^{2k}
    \;\;\; , \;\;\; k\in\N \,,
\eeqn
holds for some finite constant $M$ independent of $k$ and $t$.

Then, for $t\in[0,T)$, the estimate
\eqn 
    \tr\big( \, \big|\gamma^{(k)}(t)\, \big| \,\big)     
    \; < \; 2 \, M^{2k-2} \, (2 C M^4 T)^{(r+1)/2}   
\eeqn
holds. In particular, the right hand side converges to zero in the limit as $r\rightarrow\infty$
for $T<(2 C M^4)^{-1}$ (independent of $k$), for every $k\in\N$. 
\end{lemma}

This main lemma implies that
\eqn 
   \tr\big( \, \big|\gamma^{(k)}(t)\, \big| \,\big)  = 0  \,, \;\;\; t\in[0,T) \,,
\eeqn
and thus that $\gamma^{(k)}(t)=0$ for $t\in[0,T)$. Hence, uniqueness holds.
that the trace norm of the right hand side
converges to zero as $r\rightarrow\infty$, for $t\in[0,T)$ and $T>0$ sufficiently small (where the smallness
condition on $T$ is uniform in $k$).

A key difficulty in this approach stems from the fact that 
the interaction operator $B_{\ell+1}$ is the sum of $O(\ell)$ terms, therefore \eqref{eq-Duham-r-1-1}  
contains $O(\frac{(k+r-1)!}{(k-1)!})=O(r!)$ terms. 
The  boardgame argument allows one to control this rapid increase of 
the number of terms as $r\rightarrow\infty$, using the fact that the ordered time integrals 
$t_1\geq t_2\geq\cdots\geq t_r$ extends over a simplex of volume $O(\frac1{r!})$.
We give a short summary of the method.

Recalling that $B_{\ell+1} = \sum_{j=1}^\ell B_{j;\ell+1}$,
we write
\eqn 
J^k(\ut_{r}) = \sum_{\rho \in \cM_{k,r}}  J^k(\rho;\ut_{r}),
\eeqn
where 
\begin{align} 
    &J^k(\rho;\ut_{r}) := (i\lambda)^r U^{(k)}(t-t_1)
    B_{ \rho(k+1),k+1}U^{(k+1)}(t_1-t_2)\cdots 
    \\
    & \hspace{0.5cm}
    \cdots 
     U^{(k+\ell-1)}(t_{\ell-1}-t_{\ell})B_{\rho(k+\ell),k+\ell}\cdots
    U^{(k+r-1)}(t_{r-1}-t_{r})B_{ \rho(k+r),k+r}
    \gamma^{ (k+r)}(t_r),
    \nonumber
\end{align} 
and $\rho$ is a map   
\begin{align}\label{eq-rhodef-1}
    \rho : \{k+1,r+2,...,k+r\}&\rightarrow \{1, 2,...,k+r-1\}\,,
    \nonumber\\
    \rho(2) = 1\;\;&,\;\;\rho(j) < j \;\;\;\forall j\,.
\end{align}
Here $\cM_{k,r}$ denotes the set of all such mappings $\rho$.  

We observe that each map $\rho$ can be represented by highlighting 
one nonzero entry $B_{\rho(k+\ell),k+\ell}$ in each column of an $(k+r-1) \times r$ matrix $[A_{i,\ell}]$
with entries $A_{i,\ell}=B_{i,k+\ell}$ for $i<k+\ell$, and $A_{i,\ell}=0$ for $i\geq k+\ell$.
As an example, consider
\eqn \label{matrixB}
\left[ 
\begin{array}{cccccc}
    {\bf{B_{1,k+1}}}           & B_{1,k+2}          & ...& ... & ...& {\bf{B_{1, k+r}}} \\
    ...                                & {\bf{B_{2,k+2}}}  & ...& ... & ...& ... \\
    ...                                & ...                       & ...&  {\bf{B_{\rho(k+\ell),k+\ell}}} & ...& ... \\
    B_{k, k+1}                   & B_{k,k+2}          & ...& ... & ...& ... \\
    0                                 & B_{k+1,k+2}       & ...& ... & ...& ... \\
    ...                                & 0                        & ...& ... & ...& ... \\ 
    ...                                & ...                       & ...& ... & ...& ... \\
    ...                                & ...                       & ...& 0  & ...& ...\\
    0                                 & 0                        & ...& 0 & ...& B_{k+r-1,k+r} 
\end{array} 
\right] \,.
\eeqn
Then,
\eqn \label{muJ}
    \gamma^{(k)}(t) = \sum_{\rho \in \cM_{k,r}}
    \int_{t\geq t_1\geq \cdots\geq t_r}
     J^k(\rho,\ut_{r}) \; dt_{1} ... dt_{r}\,,
\eeqn
where the time domains are given by the same simplex $\{t>t_1>\cdots>t_r\}\subset[0,t]^r$
for all integrals in the sum over $\rho$.

We next consider the integrals with permuted time integration orders   
\eqn \label{int-musig}
    I(\rho, \pi) = \int_{ t \geq t_{\pi(1)} \geq ... \geq t_{\pi(r)} } 
    J^k(\rho;\ut_{r}) \; dt_{1} ... dt_{r},
\eeqn
where $\pi$ is a permutation of $\{1,2,...,r\}$.
This corresponds to replacing the simplex $\{t>t_1>\cdots>t_r\}\subset[0,t]^r$
by an isometric image in $[0,t]^r$.
One can associate to $I(\rho, \pi)$ the matrix 
$$ 
\left[ 
\begin{array}{cccc}
    t_{\pi^{-1}(1)} & t_{\pi^{-1}(2)} & ... & t_{\pi^{-1}(r)} \\
     {\bf{B_{1,k+1}}}     & B_{1,k+2}        & ... & {\bf{B_{1, k+r}}} \\
    ...                                & {\bf{B_{2,k+2}}} & ... & ... \\
    ...                                & ...                  & ... & ... \\
    B_{k, k+1}            & B_{k,k+2}       & ... & ... \\
    0                                 & B_{k+1,k+2}      & ... & ... \\
    ...                                 & 0      & ... & ... \\ 
    ...                                & ...                  & ... & ... \\
    0                                 & 0                    & ... & B_{k+r-1,k+r} 
\end{array} 
\right]$$
whose columns are labeled $1$ through $r$ and whose rows are labeled
$0, 1, ... ,k+r-1$, and where the highlighted entries correspond to $B_{\rho(k+\ell),k+\ell}$.

Using the combinatorial method in \cite{esy1,esy2,esy3,esy4} in the form presented in \cite{KM}, 
a {\em board game} is introduced on the set of such matrices.
A  {\em acceptable move} is characterized as follows:
If $\rho(k+\ell) < \rho(k+\ell-1)$, the player is allowed to
do the following three changes at the same time: 
\begin{itemize} 
\item exchange the highlights in columns $\ell$ and $\ell+1$,  
\item exchange the highlights in rows $k+\ell-1$ and $k+\ell$,  
\item exchange $t_{\pi^{-1}(\ell)}$ and $t_{\pi^{-1}(\ell+1)}$.
\end{itemize}
We note that the rows  $k+\ell$ and $k+\ell+1$ do not necessarily contain highlights.
A main property of the integrals $I(\rho,\pi)$ is 
{\em invariance under acceptable moves}, \cite{esy1,esy2,esy3,esy4,KM}:
 
\begin{lemma} \label{transformI}
If $(\rho,\pi)$ is transformed into $(\rho',\pi')$ by an acceptable move, 
then $I(\rho,\pi) = I(\rho',\pi')$. 
\end{lemma}

We say that a matrix of the type \eqref{matrixB} is in {\em upper echelon form}
{\em if each highlighted entry in a row is 
to the left of each highlighted entry in a lower row}. 
For example, the following matrix is in upper echelon form  (with $k=1$ and $r=4$):
$$ 
\left[ 
\begin{array}{cccc}
{\bf{B_{1,2}}} & B_{1,3}  &  B_{1,4} & B_{1,5}\\
0 & {\bf{B_{2,3}}} & B_{2,4} & B_{2,5} \\ 
0 & 0 & {\bf{B_{3,4}}} & {\bf B_{3,5}}\\
0 & 0 & 0 & B_{4,5}
\end{array} 
\right].
$$

Then, the following {\em normal form} property holds, \cite{esy1,esy2,esy3,esy4,KM}:  
\begin{lemma} 
\label{lm-echnumber-1}
For each matrix in $\cM_{k,r}$, there is a finite number of acceptable moves that transforms the matrix 
into upper echelon form. 
Moreover, let $C_{k,r}$ denote the number of upper echelon matrices of size $(k+r-1) \times r$. 
Then, 
\eqn
     C_{k,r}  \leq 2^{k+r}\,.
\eeqn 
\end{lemma}

Let   $\cN_{k,r}$ denote the subset of   matrices in $\cM_{k,r}$ which are in upper echelon form. 
Let $\sigma$ account for a matrix in $\cN_{k,r}$.  We write $\rho \sim \sigma$ 
if the matrix corresponding to $\rho$ can be transformed into that corresponding to
$\sigma$ in finitely many acceptable moves.
We note that $\sigma$ satisfies the same properties \eqref{eq-rhodef-1} as $\rho$, but in addition,
\eqn\label{eq-sigmadef-1} 
    \sigma(j)&\leq&\sigma(j') \;\;\;,\;\;\forall j<j' \,.
\eeqn  
Then, the following key theorem holds, \cite{esy1,esy2,esy3,esy4,KM}:
\begin{theorem}
\label{thm-Duhamel-iterate-bd-2} 
Suppose $\sigma \in \cN_{k,r}$. Then, there exists a subset of $[0,t]^r$, denoted by 
$D(\sigma,t)$, such that 
\eqn\label{eq-Duhamel-uppech-1}
                              \sum_{\rho \sim \sigma} \int_{t\geq t_1\geq\cdots\geq t_r}
			  J^k(\rho;\ut_{r}) \; dt_{1} ... dt_{r} 
                             = \int_{D(\sigma,t)}  J^k(\sigma;\ut_{r}) \; dt_{1} ... dt_{r} \,. 
\eeqn                          
\end{theorem} 

We remark that $D(\sigma,t)$ is  the union of all 
simplices $\{t>t_{\pi(1)}>\cdots>t_{\pi(r)}\}\subset[0,t]^r$ obtained under acceptable moves for the fixed 
upper echelon form $\sigma$;
notably, the interiors of these simplices are all pairwise disjoint.  
We emphasize that the main point of Theorem \ref{thm-Duhamel-iterate-bd-2} is the reduction of 
a sum of $O(r!)$ terms to a sum of $O(C^r)$ terms.
This concludes our summary of the Erd\"os-Schlein-Yau combinatorial method  \cite{esy1,esy2,esy3,esy4}, 
formulated in boardgame form following Klainerman-Machedon \cite{KM}.

\subsection{Setup of the proof}\label{ssec-setupproof-1}
We now give a precise formulation of the framework in which we will prove Theorem \ref{thm-main-1}.
Let us assume that we have two positive semidefinite
solutions $(\gamma_j^{(k)}(t))_{k\in\N}\in L^\infty_{t\in[0,T)}\frH^1$ satisfying the same initial data,
$(\gamma_1^{(k)}(0))_{k\in\N}=(\gamma_2^{(k)}(0))_{k\in\N}\in\frH^1$. 
Then, 
\eqn 
    \gamma^{(k)}(t) := \gamma_1^{(k)}(t) - \gamma_2^{(k)}(t)
    \;\;\;,\;\;\;k\in\N\,,
\eeqn 
is a solution to the GP hierarchy  
with initial data $\gamma^{(k)}(0)=0$ $\forall k\in\N$,
and it suffices to prove that $\gamma^{(k)}(t)=0$ $\forall k\in\N$, and for all $t\in[0,T)$.
This is due to the linearity of the GP hierarchy.

We note that $\gamma^{(k)}$, as a difference of positive semidefinite marginal density matrices,
is not in general positive semidefinite.

From the assumptions of Theorem \ref{thm-main-1}, we have that
\eqn\label{eq-Skgamma-cond-1}
    \sup_{t\in[0,T)}\tr(|S^{(k,1)}[\gamma_i^{(k)}(t)] |) < M^{2k}
    \;\;\; , \;\;\; k\in\N \;\;,\;i=1,2,
\eeqn
for some finite constant $M$ independent of $k$ and $t$.

To ensure the applicability of Theorems \ref{thm-strongDeFinetti-1} and \ref{thm-weakDeFinetti-1}, we note that if  
\eqn
    \gamma_j^{(k)}=\tr_{k+1}(\gamma_j^{(k+1)}) 
    \;\;\; \forall k\in\N \;\;,\;j=1,2 \,,
\eeqn
are admissible, it follows immediately that $(\gamma^{(k)})_{k\in\N}$ is admissible.
Moreover if both $(\gamma_1^{(k)})$ and $(\gamma_2^{(k)})$ are obtained from a weak-* limit,
then so is $(\gamma^{(k)})$.

Thus, from Theorems \ref{thm-strongDeFinetti-1} and \ref{thm-weakDeFinetti-1}, we have that 
\eqn
    \gamma_j^{(k)}(t) = \int  d\mu_t^{(j)}(\phi)\big( |\phi\rangle\langle \phi |\big)^{\otimes k} 
    \;\;\;,\;\;\;j=1,2\,,
    \nonumber\\
    \gamma^{(k)}(t) = \int  d\widetilde\mu_t(\phi)\big( |\phi\rangle\langle \phi |\big)^{\otimes k} \,,
\eeqn
where $\widetilde{\mu}_t:=\mu^{(1)}_t-\mu^{(2)}_t$ is the difference of two probability measures 
on the unit ball in $L^2(\R^3)$.
We remark that \eqref{eq-Skgamma-cond-1} is equivalent to  
\eqn\label{eq-muH1bound-1}
     \int  d\mu_t^{(j)}(\phi) \|\phi\|_{H^1}^{2k} < M^{2k}
     \;\;\;,\;\;\;j=1,2\,,
\eeqn
for all $k\in\N$, where $H^1=\{f\in L^2(\R^3) \, | \, \|\langle\nabla_x\rangle f\|_{L^2}<\infty\}$,
and $\langle\nabla\rangle:=\sqrt{1-\Delta}$.

\begin{lemma}\label{lm-Chebyshev-1}
Let $\mu$ be a Borel probability measure in $L^2(\R^3)$, and assume that 
\eqn\label{eq-EnCond-1}
    \int d\mu(\phi) \|\phi\|_{H^1}^{2k} \leq M^{2k}
\eeqn
holds for some finite constant $M>0$, and all $k\in\N$. 
Then, 
\eqn
    \mu\Big(\Big\{ \, \phi\in L^2(\R^3) \, \Big| \,   \|\phi\|_{H^1} >  M \,\Big\} \Big)\, = \, 0 \,.
\eeqn
\end{lemma}

\begin{proof}
From Chebyshev's inequality, we have that
\eqn
    \mu\Big(\Big\{  \phi\in L^2(\R^3) \Big| \,   \|\phi\|_{H^1} > \lambda \Big\}\Big)
    \leq\frac1{\lambda^{2k}}\int d\mu(\phi) \|\phi\|_{H^1}^{2k} \leq \frac{M^{2k}}{\lambda^{2k}} \,
\eeqn
for any $k\geq0$.
Evidently, for $\lambda> M$, the r.h.s. tends to zero when $k\rightarrow\infty$.
\end{proof}

From here on, we will consider the representation of the expansion  \eqref{eq-Duham-r-1-1}  for $\gamma^{(k)}(t)$
in upper echelon normal form, given by the right hand side of \eqref{eq-Duhamel-uppech-1}.
Then,
\eqn  
    \gamma^{(k)}(t) 
    &=&\sum_{\sigma\in\cN_{k,r}} 
    \int_{D(\sigma,t)} dt_{1}\cdots dt_r   U^{(k)}(t-t_1)
    B_{\sigma(k+1),k+1}U^{(k+1)}(t_1-t_2)\cdots 
    \nonumber\\
    &&\hspace{2cm}
    \cdots
    U^{(k+r-1)}(t_{r-1}-t_{r})B_{\sigma(k+r),k+r}
    \gamma^{(k+r)}(t_r)\,
    \label{eq-DuhamEchelon-def-1-0} \,.
\eeqn
The sum with respect to $\sigma$ extends over all inequivalent upper echelon forms.  

Using the quantum de Finetti theorem (by which we henceforth refer to either the strong or the weak version), we obtain
\eqn\label{eq-DuhamEchelon-def-1-1} 
    \gamma^{(k)}(t) 
    &=&\sum_{\sigma\in\cN_{k,r}} \int_{D(\sigma,t)} dt_1,\dots,dt_r  
     \int d\widetilde{\mu}_{t_r}(\phi) \, \opJ^k(\sigma;t,t_1,\dots,t_r) \,,
\eeqn
where
\eqn
    \lefteqn{
    \opJ^k(\sigma;t,t_1,\dots,t_r;\ux_k;\ux_k')
    =\Big( \, U^{(k)}(t-t_1)
    B_{\sigma(k+1),k+1}U^{(k+1)}(t_1-t_2)\cdots 
    }
    \nonumber\\
    &&\hspace{0.5cm}
    \cdots B_{\sigma(k+\ell),k+\ell}U^{(k+\ell)}(t_{\ell}-t_{\ell+1})
    B_{\sigma(k+\ell+1),k+\ell+1}U^{(k+\ell+1)}(t_{\ell+1}-t_{\ell+2})\cdots
    \nonumber\\
    &&\hspace{1cm}
    \cdots
    U^{(k+r-1)}(t_{r-1}-t_{r})B_{\sigma(k+r),k+r}
    \big( |\phi\rangle\langle \phi |\big)^{\otimes (k+r)} \,\Big)(\ux_k;\ux_k') \,.
    \label{eq-DuhamEchelon-def-1}
\eeqn
Here, we may think of the time variable $t_\ell$ as being attached to the interaction operator $B_{\sigma(k+\ell),k+\ell}$.
For {\em fixed} $\phi$, we note that since 
\eqn\label{eq-gammaphi-prod-def-1}
    \big( |\phi\rangle\langle \phi |\big)^{\otimes (k+r)}(\ux_{k+r};\ux_{k+r}') 
    =\prod_{i=1}^{k+r}( |\phi\rangle\langle \phi |)(x_i;x_i')
\eeqn
is given by a product of 1-particle kernels, it follows that
\eqn\label{eq-gammak-prod-def-1}
    \opJ^k(\sigma;t,t_1,\dots,t_r;\ux_k;\ux_k') = \prod_{j=1}^k 
    \opJ_j^1(\sigma_j;t,t_{\ell_{j,1}},\dots,t_{\ell_{j,m_j}};x_j;x_j')
\eeqn 
likewise has product form, for each fixed $\sigma$. 
This is because in \eqref{eq-DuhamEchelon-def-1}, the 
operators $B_{\sigma(k+\ell),k+\ell}$ and $U^{(k+\ell)}(t_{\ell}-t_{\ell+1})$ each map products of 1-particle
kernels to products of 1-particle kernels (but the operators $B_{\sigma(k+\ell),k+\ell}$ do in general
not preserve positive semidefiniteness).
Each 1-particle kernel $\opJ_j^1$ can be written as a Duhamel expansion in itself, with interaction operators
inherited from those appearing in  $\opJ^k$. We label the interaction operators in  $\opJ_j^1$ 
``internally"   
with $\sigma_j$, $j=1,\dots,k$,  (which are automatically in upper echelon form relative to $\opJ_j^1$). 
More details are given in Section \ref{sec-treegraph-1} below.

For a fixed $k$, the number of  inequivalent echelon forms  is bounded by $C^r$,  using Lemma 
\ref{lm-echnumber-1}. Hence,
\eqn\label{eq-gammak-L2bound-1}
    \lefteqn{
    \tr( \, |\gamma^{(k)}| \, ) 
    }
    \\
    &\leq& C^r \sum_{i=1,2}
    \sup_\sigma \int_{[0,t]^r}dt_1\cdots dt_r 
    \int d{\mu}_{t_r}^{(i)}(\phi) 
    \prod_{j=1}^k\tr\Big( \, \Big| \,\opJ_j^1(\sigma_j;t,t_{\ell_{j,1}},\dots,t_{\ell_{j,m_j}}) \, \Big| \, \Big)\,.
    \nonumber
\eeqn  
The time variable $t_{\ell_{j,\alpha}}$ corresponds to the one attached to the $\alpha$-th interaction operator
(counting from the left) appearing in the factor $\opJ_j^1$ (every  $t_{\ell_{j,\alpha}}$ corresponds uniquely to one of the 
time variables $t_\ell$ in \eqref{eq-DuhamEchelon-def-1}).
Based on the expression \eqref{eq-gammak-L2bound-1}, our goal is to prove the estimate
\eqn 
    \tr\big( \, \big|\gamma^{(k)}(t)\, \big| \,\big)     
    \; < \; 2 \, M^{2k-2} \, (2 C M^4 T)^{(r+1)/2}   
\eeqn
asserted in Lemma \ref{lm-mainlm-trest-1}; it implies 
that for any $k\in\N$,
the right hand side tends to zero as $r\rightarrow\infty$, for $t\in[0,T)$, and sufficiently small $T>0$
(independent of $k$).
Since $r$ is arbitrary, this implies that the left hand side equals zero, thus establishing uniqueness.
By iterating this argument on the union of intervals $[0,T)\cup[T,2T)\cup\dots$, uniqueness extends to the 
entire time of existence for a given solution.
We note that in \eqref{eq-gammak-L2bound-1}, the distinction between focusing and defocusing GP hierarchy has 
disappeared, since $|\lambda|=1$ in both cases.

\section{Binary tree graphs}
\label{sec-treegraph-1}

Because every interaction operator contracts precisely two factors, one can conveniently organize the above
expansions for $\opJ^k$ and $\opJ_j^1$ with the help of binary tree graphs, 
for arbitrary values of $k$ and $r$.  
For the convenience of the reader, we first discuss the factorization \eqref{eq-gammak-prod-def-1} with an example. 

\subsection{An example for $k=3$}
\label{ssec-k3trees-ex-1}
As an example, let $k=3$ and $r=4$, and let us consider
\eqn\label{eq-Iexample34-1}
    \opJ^3(\sigma;t,t_1,\dots,t_4) = U_{0,1}^{(3)}B_{2,4}U_{1,2}^{(4)}B_{2,5}
    U_{2,3}^{(5)}B_{3,6}U_{3,4}^{(6)}B_{5,7}(|\phi\rangle\langle\phi|)^{\otimes 7}
\eeqn
where  $U^{(j)}_{i,i'}:=U^{(j)}(t_{i}-t_{i'})$ with $t_0:=t$,
and where $\sigma$ corresponds to the upper echelon matrix
\eqn \label{matrixB-ex-2}
\left[ 
\begin{array}{cccccc}
    B_{1,4}                    & B_{1,5}          & B_{1,6}          &B_{1,7}  \\
    {\bf{B_{2,4}}}          & {\bf{B_{2,5}}}   & B_{2,6}          &B_{2,7} \\
    B_{3,4}                   & B_{3,5}           & {\bf B_{3,6} }  &B_{3,7} \\
    0                             & B_{4,5}           & B_{4,6}         &B_{4,7} \\
    0                             & 0                    & B_{5,6}          &{\bf B_{5,7}} \\ 
    0                             & 0                    & 0                    &B_{6,7} 
\end{array} 
\right] \,.
\eeqn
We now read off from this matrix which terms are ``connected" via contractions, starting with the rightmost interaction operator.
Although all 7 factors in $(|\phi\rangle\langle\phi|)^{\otimes 7}$ are indistinguishable, we enumerate
the factors, and write the product 
in the form $\otimes_{i=1}^7u_i$, ordered with increasing index $i$
(where $u_i=|\phi\rangle\langle\phi|$ for every $i=1,\dots,7$). 
\begin{itemize}
\item
Clearly, $B_{5,7}$ contracts the factors $u_5$ and $u_7$, and acts trivially as the identity on all other factors $u_i$,
\eqn\label{eq-ex-B57term-1}
    B_{5,7}(\otimes_{i=1}^7u_i ) = (\otimes_{i=1}^4 u_i) \otimes \Theta_4 \otimes u_6\,,
\eeqn
where
\eqn
     \Theta_4:=B_{1,2}(u_5\otimes u_7) \,.
\eeqn
The index $\alpha$ in $\Theta_\alpha$ associates it to the 
$\alpha$-th interaction operator from the left in \eqref{eq-Iexample34-1} 
(in case of $\Theta_4$, the 4th interaction operator is given by $B_{5,7}$). 
\item
The interaction operator $B_{3,6}$ contracts $U^{(1)}_{3,4}u_3$ and $U^{(1)}_{3,4}u_6$, while it leaves
all remaining factors untouched. In particular, it does not affect $\Theta_4$.
\eqn\label{eq-ex-B36term-1}
    B_{3,6}U_{3,4}^{(6)}\Big(\eqref{eq-ex-B57term-1}\Big)
    = (U_{3,4}^{(2)}(u_1\otimes u_2)) \otimes \Theta_3\otimes  (U_{3,4}^{(1)}u_4)\otimes (U_{3,4}^{(1)}\Theta_4 )
\eeqn
where
\eqn
    \Theta_3:=B_{1,2}(U_{3,4}^{(2)}(u_3\otimes u_6)) \,.
\eeqn
\item
The interaction operator $B_{2,5}$ contracts $U^{(1)}_{2,4}u_2$ with $U^{(1)}_{2,4}\Theta_4$
(where we used the group property $U^{(j)}_{2,3}U^{(j)}_{3,4}=U^{(j)}_{2,4}$)
corresponding to the 2nd and 5th factor in \eqref{eq-ex-B36term-1}, while it leaves 
all remaining factors untouched.
\eqn\label{eq-ex-B25term-1}
    B_{2,5}U_{2,3}^{(5)}\Big(  \eqref{eq-ex-B36term-1} \Big)
    =  (U_{2,4}^{(1)}u_1) \otimes \Theta_2 \otimes (U^{(1)}_{2,3}\Theta_3)\otimes  (U_{2,4}^{(1)}u_4)
\eeqn
where 
\eqn
    \Theta_2:= B_{1,2}\big((U_{2,4}^{(1)}u_2)\otimes (U_{2,4}^{(1)}\Theta_4)\big)\,.
\eeqn
\item
Finally, the interaction operator $B_{2,4}$ contracts $U_{1,2}^{(1)}\Theta_2$ with $(U_{1,4}^{(1)}u_4)$
corresponding to the 2nd and 4th factor in \eqref{eq-ex-B25term-1},
while leaving all other factors untouched,
\eqn
    B_{2,4}U_{1,2}^{(4)}\Big(  \eqref{eq-ex-B25term-1} \Big) 
    =   (U_{1,4}^{(1)}u_1) \otimes \Theta_1 \otimes (U^{(1)}_{1,3}\Theta_3) 
\eeqn
where
\eqn
    \Theta_1:=B_{1,2}\big((U_{1,2}^{(1)}\Theta_2)\otimes (U_{1,4}^{(1)}u_4)\big)\,.
\eeqn
\end{itemize}
In conclusion, we have found the factorized expression for \eqref{eq-Iexample34-1},  
\eqn
    \opJ^3 =  \underbrace{(U_{0,4}^{(1)}u_1)}_{=\opJ_1^1} \otimes 
    \underbrace{(U_{0,1}^{(1)}\Theta_1)}_{=\opJ_2^1} \otimes 
    \underbrace{(U^{(1)}_{0,3}\Theta_3)}_{=\opJ_3^1} \,.
\eeqn
We may now write the factors $\opJ_j^1$ as one-particle matrices, and substitute back $u_i=|\phi\rangle\langle\phi|$ for all $i=1,\dots,7$.
\begin{itemize}
\item
$\opJ_1^1$ corresponds to a free propagation without any interaction operators,
\eqn
    \opJ_1^1 = U_{0,4}^{(1)}|\phi\rangle\langle\phi|
\eeqn
\item
Moreover,
\eqn
    \opJ_2^1 = U_{0,1}^{(1)}B_{1,2}U^{(2)}_{1,2}B_{1,3}U^{(1)}_{2,4}B_{3,4}(|\phi\rangle\langle\phi|)^{\otimes4}
\eeqn
where the interaction operators correspond to $B_{2,4}$, $B_{2,5}$, $B_{5,7}$; they are re-indexed in a manner
that leaves the connectivity structure between contractions invariant.
The labeling of interaction operators $B_{\sigma_2(\ell),\ell}$ 
here is obtained from a labeling function $\sigma_2$ (corresponding to $\sigma_j$ in 
\eqref{eq-gammak-prod-def-1}) where $\sigma_2(2)=1$, $\sigma_2(3)=1$, $\sigma_2(4)=3$.
\item
Finally,
\eqn
    \opJ_3^1 = U_{0,3}^{(1)}B_{1,2}U^{(2)}_{3,4} (|\phi\rangle\langle\phi|)^{\otimes2}
\eeqn
where the interaction operator corresponds to $B_{3,6}$, and can be labeled with $\sigma_3(2)=1$.
\end{itemize}
We observe that for $\ell<\ell'$, the interaction operators $B_{\sigma(\ell),\ell}$ 
and $B_{\sigma(\ell'),\ell'}$ in $\opJ^3$ (which are highlighted in \eqref{matrixB-ex-2}) 
belong to the same factor $\opJ_j^1$ if either $\sigma(\ell)=\sigma(\ell')$, or $\ell=\sigma(\ell')$.
In this case, we can think of them as being ``connected"; below, we will introduce binary tree graphs
that encode this connectivity structure.

This example also illustrates how the internal labeling functions $\sigma_j$ in \eqref{eq-gammak-prod-def-1} 
are deduced from the ``global" labeling function $\sigma$. In the sense outlined above, we may think of
$\sigma_j$ as the restriction of $\sigma$ to $\opJ_j^1$.

We note that in $\opJ_1^1$ and $\opJ_3^1$, there is a free propagator applied to each $\phi$;
this will allow for a straightforward application of Strichartz estimates to control contractions due to interaction operators.
However, the term $\opJ_2^1$  involves the contraction of factors $|\phi\rangle\langle\phi|$ without 
any free propagation term inbetween.  
We will call  the factor $\opJ_2^1$ {\em distinguished} while $\opJ_1^1$ and $\opJ_3^1$ are regular.

\subsection{Definition of binary trees}
\label{subsec-trees-def-1}

We now introduce binary tree graphs as a bookkeeping  device to keep track of the complicated
contraction structures  imposed by the interaction operators inside  the
iterated Duhamel formula
\eqref{eq-DuhamEchelon-def-1}.

To this end, we
associate  \eqref{eq-DuhamEchelon-def-1} to the union of $k$ disjoint binary tree graphs, $(\tau_j)_{j=1}^k$.
We note that these appear as ``skeleton graphs" for the more complicated
graphs in \cite{esy1,esy2,esy3,esy4}.
We assign:
\begin{itemize}
\item 
An {\em internal vertex} $v_\ell$, $\ell=1,\dots,r$, to each operator $B_{\sigma(k+\ell),k+\ell}$.
Accordingly, the time variable $t_\ell$ in \eqref{eq-DuhamEchelon-def-1}
is thought of as being attached to the vertex $v_\ell$.
\item
A {\em root vertex} $w_j$, $j=1,\dots,k$ 
to each factor $\opJ_j^1(\cdots;x_j;x_j')$ in \eqref{eq-gammak-prod-def-1}.
\item
A {\em  leaf vertex} $u_i$, $i=1,\dots,k+r$, to the factor $( |\phi\rangle\langle \phi |)(x_i;x_i')$ 
in \eqref{eq-gammaphi-prod-def-1}.
\end{itemize}
For the sake of concreteness, we draw graphs as follows:
We consider the strip in $(x,y)\in\R^2$ given by $x\in[0,1]$.
We draw all root vertices $(w_j)_{j=1}^k$, ordered vertically,   on the line $x=0$,
all internal vertices $(v_\ell)_{\ell=1}^r$ in the region $x\in(0,1)$, where $v_{\ell'}$ is on the right of 
$v_\ell$ if $\ell'>\ell$. Finally, we draw all leaf vertices $(u_i)_{i=1}^{k+r}$, ordered vertically, on the line
$x=1$.   

Next, we introduce the equivalence relation ``$\sim$" of {\em connectivity} between vertices. Between any 
pair of connected vertices, we draw a connecting line, which we refer to as an {\em edge}:
\begin{itemize}
\item 
Let $v_\ell$ be the internal vertex with smallest value of $\ell$ such that $\sigma(\ell)=j$;
then, we say that $v_\ell$ is connected to the root vertex $w_j$,  that is, $w_j\sim v_\ell$.
\item 
If there is no internal vertex connected to $w_j$, we draw an edge connecting $w_j$ to the
leaf vertex $u_j$, and say that they are connected, $w_j\sim u_j$.
\item 
Given $k<\ell\leq k+r$, if there exists $\ell'>\ell$ such that
$\ell=\sigma(\ell')$ or $\sigma(\ell)=\sigma(\ell')$,  we say that $v_\ell\sim v_{\ell'}$ are {\em connected}.

$v_\ell$ is then called a {\em parent vertex} of $v_{\ell'}$, and 
$v_{\ell'}$ is called a {\em child vertex} of $v_\ell$. We denote the two child vertices of $v_\ell$
by $v_{\kappa_-(\ell)}$ and $v_{\kappa_+(\ell)}$, using the condition
$\kappa_-(\ell)<\kappa_+(\ell)$.

If there exists no internal vertex $v_{\ell'}$ with $\ell'>\ell$ such that $\ell=\sigma(\ell')$,
we say that $v_\ell$ is connected to the leaf vertex $u_\ell$, $v_\ell\sim u_\ell$; if there exists
no  internal vertex $v_{\ell'}$ with $\ell'>\ell$ such that $\sigma(\ell)=\sigma(\ell')$,
we say that $v_\ell$ is connected to the leaf vertex $u_\sigma(\ell)$, $v_\ell\sim u_{\sigma(\ell)}$.
In these cases, $v_\ell$ is the {\em parent vertex} of $u_\ell$ (or $u_{\sigma(\ell)}$), and 
$u_\ell$ (or $u_{\sigma(\ell)}$) is a {\em child vertex} of $v_\ell$.
\end{itemize}

This implies that every internal vertex has precisely two child vertices, which can
either be internal or leaf vertices (they do not need to be of the same type).
Every root vertex has precisely one child vertex, which could be of internal or of leaf type.
Every internal or leaf vertex has exactly one parent vertex.

We conclude that the graph thus obtained is the disjoint union of $k$ binary trees, which we denote
by $(\tau_j)_{j=1}^k$, where the root of $\tau_j$ is the root vertex $w_j$ (if $w_j\sim u_j$ without
internal vertices inbetween, then the
binary tree consists trivially only of a single edge connecting one root and one leaf vertex).

We say that the tree $\tau_j$ is {\em distinguished} if $v_r\in\tau_j$,
and {\em regular} if $v_r\not\in\tau_j$. We call the two leaf vertices connected to $v_r$ {\em distinguished
leaf vertices}, and all others {\em regular leaf vertices}.
Clearly, there are $k-1$ regular trees, and one distinguished tree in this construction.\noindent

\centerline{\epsffile{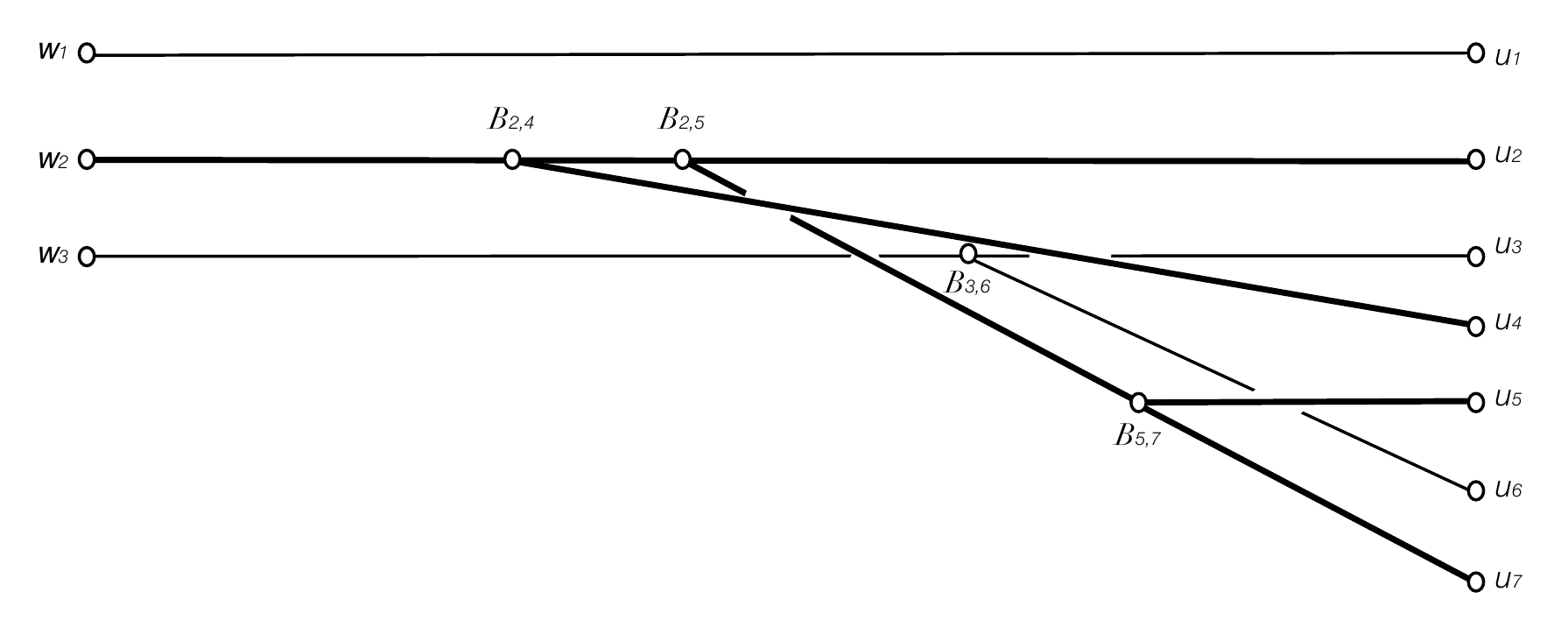} }
\noindent
Figure 1. The disjoint union of three tree graphs $\tau_j$, $j=1,2,3$, 
corresponding to the example discussed in Section \ref{ssec-k3trees-ex-1},
for $k=3$, $r=4$, and
\eqn\label{eq-Iexample34-2}
    \opJ^3(\sigma;t,t_1,\dots,t_4) = U_{0,1}^{(3)}B_{2,4}U_{1,2}^{(4)}B_{2,5}
    U_{2,3}^{(5)}B_{3,6}U_{3,4}^{(6)}B_{5,7}(|\phi\rangle\langle\phi|)^{\otimes 7} \,.
\eeqn 
The root vertex $w_j$ belongs to the tree $\tau_j$, $j=1,2,3$.
The internal vertices correspond to $v_1\sim B_{2,4}$, $v_2\sim B_{2,5}$, $v_3\sim B_{3,6}$, and $v_4\sim B_{5,7}$.
The leaf vertices $u_5$ and $u_7$, and the internal vertex  $v_4\sim B_{5,7}$ are distinguished.
The distinguished tree $\tau_2$ is drawn with thick edges.

\section{The distinguished tree graph}
\label{sec-disttree-1}

In this section, we further refine the combinatorial organization of terms corresponding to the
distinguished tree. We note that regular trees can be treated in a similar manner, with obvious modifications.
Let $\tau_j$ denote the distinguished tree graph.
We assume that it
contains $m_j$ internal vertices $(v_{\ell_{j,\alpha}})_{\alpha=1}^{m_j}$, and $\Lj$ leaf  vertices 
$(u_{j,i})_{i=1}^{\Lj}$.
We recall that the internal vertices are enumerated with $\alpha\in\{1,\dots,m_j\}$,
where $\alpha=m_j$ is the distinguished vertex, and that $\alpha$ corresponds to the
interaction operator $B_{\sigma_j(\alpha+1),\alpha+1}$ (notice the shift by 1 in the index) in \eqref{eq-Ijdisting-def-1}.
For notational simplicity, we will from here on label leaf vertices with $\alpha\in\{m_j+1,\dots,2m_j+2\}$
(corresponding to $u_{j,\alpha-m_j}$), 
and will often refer to the vertex $v_{j,\alpha}$ by its label $\alpha$.

To determine the contribution to \eqref{eq-gammak-prod-def-1} corresponding to  $\tau_j$,
we use the commutativity relation \eqref{eq-UBcommut-1}, 
and straightforwardly find that
\begin{align}\label{eq-Ijdef-0-1}
    &\opJ_j^1(\sigma_j;t,t_{\ell_{j,1} },\dots,t_{\ell_{j,m_j}};x_j;x_j')
    \\
    &=
    U^{(1)}(t-t_{\ell_{j,1} })\cdots U^{(1)}(t_{\ell_{j,1}-1}-t_{\ell_{j,1}})
    B_{\sigma_j(2),2}\cdots 
    \nonumber\\
    &\hspace{0.5cm}
    \cdots B_{\sigma_j(\alpha),\alpha}U^{(\alpha)}(t_{\ell_{j,\alpha-1}}-t_{\ell_{j,\alpha-1}+1})
    \cdots U^{(\alpha)}(t_{\ell_{j,\alpha}-1}-t_{\ell_{j,\alpha}})
    B_{\sigma_j(\alpha+1),\alpha+1}\cdots 
    \nonumber\\
    &\hspace{1cm}
    \cdots
    U^{(m_j) }(t_{\ell_{j,m_j}-1}-t_{\ell_{j,m_j}})B_{\sigma_j(m_j+1),m_j+1} 
    \big( |\phi\rangle\langle \phi |\big)^{\otimes (\Lj)} \,.
    \nonumber 
\end{align}
Here, the interaction operators $B_{\sigma_j(\alpha),\alpha}$ are adapted to  the
1-particle   kernel $\opJ_j^1$, and   $(\sigma_j)_{j=1}^{m_j}$ is an ``internal" labeling
of the interaction operators that preserves  
the structure of $\tau_j$. In this sense, $\sigma_j$ corresponds to
the restriction of $\sigma$ to the tree $\tau_j$. Clearly, $\sigma_j(2)=1$.

Between any two consecutive interaction operators 
$B_{\sigma_j(\alpha),\alpha}$
and $B_{\sigma_j(\alpha+1),\alpha+1}$,
with $\alpha<m_j$, there is a composition of $\ell_{j,\alpha}-\ell_{j,\alpha-1}$
free propagators at consecutive time steps, so that
\eqn 
    U^{(\alpha)}(t_{\ell_{j,\alpha-1}}-t_{\ell_{j,\alpha-1}+1})
    \cdots U^{(\alpha)} (t_{{\ell_{j,\alpha}}-1}-t_{\ell_{j,\alpha}})
    = U^{(\alpha)} (t_{\ell_{j,\alpha-1}}-t_{\ell_{j,\alpha}} ) \,,
\eeqn
due to the group property of the free propagators. 
Hence,     \eqref{eq-Ijdef-0-1} reduces to 
\eqn\label{eq-Ijdisting-def-1}
    \lefteqn{
    \opJ_j^1(\sigma_j;t,t_{\ell_{j,1}},\dots,t_{\ell_{j,m_j} })=U^{(1)}(t-t_{\ell_{j,1} })
    B_{1,2}\cdots 
    }
    \\
    &&
    \cdots B_{\sigma_j(\alpha),\alpha}U^{(\alpha)}(t_{\ell_{j,\alpha-1}}-t_{\ell_{j,\alpha}})
    B_{\sigma_j(\alpha+1),\alpha+1}\cdots 
    \nonumber\\
    &&\hspace{1cm}
    \cdots
    U^{(m_j)}(t_{\ell_{j,m_j-1}}-t_{\ell_{j,m_j}})B_{\sigma_j(m_j+1),m_j+1} 
    \big( |\phi\rangle\langle \phi |\big)^{\otimes (\Lj) }
    \nonumber
\eeqn
where $\ell_{j,m_j}=r$. 
We observe that on the last line, there is no free propagator in front of 
$\big( |\phi\rangle\langle \phi |\big)^{\otimes (\Lj)}$ (since $\tau_j$ is distinguished).
As a consequence, the Strichartz estimate  cannot be 
applied in the last step. 
Resolving this issue is the main task of
the construction presented in the sequel.

Our goal is to bound
\eqn 
     \int_{[0,T)^{m_j-1}}  dt_{\ell_{j,1}}\cdots dt_{\ell_{j,m_j-1} }
     \tr\Big( \, \Big| \, \opJ_j^1(\sigma_j;t,t_{\ell_{j,1}},\dots,t_{\ell_{j,m_j} }) \, \Big| \,\Big) \, .
\eeqn 
We will employ a recursion that takes into account the structure
of interactions and free evolutions occurring between interactions.
We will explain the strategy based on an example in the next section.

\subsection{Example calculation for a distinguished tree} 
\label{ssec-disttreecalc-ex-1}
We consider an example of  a distinguished tree, which we obtain from setting $k=1$, $r=3$
in \eqref{eq-DuhamEchelon-def-1-1} (if $k=1$, there is only one tree, and it is necessarily distinguished).
From \eqref{eq-DuhamEchelon-def-1-1} and \eqref{eq-DuhamEchelon-def-1}, we have
\eqn 
    \gamma^{(1)}(t)
    &=&(i\lambda)^3\sum_{\sigma\in\cN_{1,3}} 
    \int_{D(\sigma,t)} dt_{1}dt_2 dt_3 \int d\widetilde{\mu}_{t_3}(\phi) U^{(1)}(t-t_1)
    B_{1,2}U^{(2)}(t_1-t_2)\ 
    \nonumber\\
    &&\hspace{0.5cm}
     B_{\sigma(3),3}U^{(3)}(t_{2}-t_{3})
    B_{\sigma(4),4}  
    \big( |\phi\rangle\langle \phi |\big)^{\otimes 4}   \,
    \label{eq-DuhamEchelon-ex-1} \,,
\eeqn
where $D(\sigma,t)\subseteq[0,t]^3$. For a fixed $\sigma$ (with, say, $\sigma(3)=2$ and $\sigma(4)=3$), 
we consider, as an example, the contribution
to the bound \eqref{eq-gammak-L2bound-1} of the form
\begin{align}
\label{eq-ex-bound-1} 
    &\int_{[0,T)^3}dt_1 dt_2 dt_3
    \int d{\mu}_{t_3}^{(i)}(\phi) \tr\Big( \, \Big| \, \Big( U^{(1)}(t-t_1)
    B_{1,2}U^{(2)}(t_1-t_2)\  
    \nonumber\\
    & \hspace{3cm}
     B_{2,3}U^{(3)}(t_{2}-t_{3})
    B_{3,4}  
    \big( |\phi\rangle\langle \phi |\big)^{\otimes 4} \Big|\, \Big)\,,
\end{align}
where $t\in[0,T)$, and noting that $|i\lambda|=1$.

\subsubsection{Recursive determination of contraction structure}

Clearly, $\big( |\phi\rangle\langle \phi |\big)^{\otimes 4} $ is a product of 1-particle density matrices.
We observe that the interaction operators $B_{i,j}$ preserve the product structure 
(while changing the explicit expressions for each factor), and contract 
two factors at a time (the $i$-th and the $j$-th). On all other factors,  $B_{i,j}$ acts as the identity.
Similarly as in the example of Section \ref{ssec-k3trees-ex-1}, we introduce kernels $\Theta_\alpha$, $\alpha=1,\dots,3$ 
that account for the contractions performed by $B_{\sigma(\alpha+1),\alpha+1}$, which we write in the
normal form 
\eqn\label{eq-Thetaalpha-ex-def-1}
    \Theta_\alpha(x,x')=\sum_{\beta_\alpha} c_{\beta_\alpha}^\alpha \chi_{\beta_\alpha}^\alpha(x)\overline{\psi_{\beta_\alpha}^\alpha}(x')
\eeqn
where $\chi_{\beta_\alpha}^\alpha$, $\psi_{\beta_\alpha}^\alpha$ are certain functions that will be recursively determined, and $c_{\beta_\alpha}^\alpha$
are coefficients with values in $\{1,-1\}$.

\noindent
$\bullet$ {\em The kernel $\Theta_3$:}
We start at the last interaction operator $B_{3,4}$ in \eqref{eq-ex-bound-1}. 
It acts nontrivially only on the 3-rd and 4-th factor in $(|\phi\rangle\langle\phi|)^{\otimes 4}$,
\eqn\label{eq-B34-ex-1}
    B_{3,4}(|\phi\rangle\langle\phi|)^{\otimes 4} = (|\phi\rangle\langle\phi|)^{\otimes 2}\otimes\Theta_3 \,.
\eeqn
The kernel $\Theta_3$ is obtained from contracting a two particle density matrix to a one particle density matrix via 
the interaction operator $B_{1,2}$ (which acts on a two-particle kernel $f(x,y;x',y')$ by $(B_{1,2}f)(x,x')=f(x,x;x',x)-f(x,x';x',x')$),
\eqn\label{eq-ex-Theta3-def-1}
    \Theta_3(x,x')&:=&B_{1,2}\Big(\big( |\phi\rangle\langle \phi |\big)^{\otimes 2}\Big)(x,x')
    =\psidistex(x)\overline{\phi(x')}-\phi(x)\overline{\psidistex(x')}
    \nonumber\\
    &=:&\sum_{\beta_3=1}^2 c_{\beta_3}^3 \chi_{\beta_3}^3(x)\overline{\psi_{\beta_3}^3}(x')
\eeqn
where  
\eqn\label{eq-psir-def-1}
    \psidistex:= |\phi|^2\phi \,.
\eeqn
Here, we have $c_1^3=1$, $c_2^3=-1$, $\chi_1^3=\psidistex$, $\chi_2^3=\phi$, $\psi_1^3=\phi$, $\psi_2^3=\psidistex$.

\noindent
{\bf Main difficulty:} 
{\em 
The main difficulty  in estimating \eqref{eq-ex-bound-1} 
stems from the fact that the term $\psidistex=|\phi|^2\phi$ can only be controlled in $L^2$, where by 
Sobolev embedding, $\|\psidistex\|_{L^2}\leq C \|\phi\|_{H^1}^3$, which can then
be controlled by \eqref{eq-muH1bound-1}, see \eqref{eq-ex-laststep-1} below. 
Our objective thus is to apply the triangle inequality to the trace norm  inside  \eqref{eq-ex-bound-1},
and  to    recursively ``propagate" the resulting $L^2$ norm  
through all intermediate terms until we reach $\psidistex$, see \eqref{eq-ex-bound-2} below.
We remark that if $\|\psidistex\|_{H^1}$  could be controlled by $\|\phi\|_{H^1}$ (which is not the case), 
a straightforward application of the method of  \cite{KM} would suffice 
to carry out our analysis.
}

We now re-interpret $\psidistex$ in \eqref{eq-ex-Theta3-def-1} as a function 
that is {\em independent} of $\phi$, $\overline{\phi}$.
Only at the end of our analysis, we will substitute $\psidistex:= |\phi|^2\phi$.  
We call a factor $\chi_{\beta_\alpha}^\alpha$, $\psi_{\beta_\alpha}^\alpha$ in the sum \eqref{eq-Thetaalpha-ex-def-1} 
{\em distinguished} if it is a function of $\psidistex$.
In the first step, it is clear that for every $\beta_3$, only one out of the two factors $\chi_{\beta_3}^3$, $\psi_{\beta_3}^3$
in the sum \eqref{eq-ex-Theta3-def-1} is distinguished (and in fact equal to $\psidistex$). The property of being distinguished then propagates from there, i.e., in the next step the distinguished term is the one containing the distinguished term as a factor 
from the previous step, and so on.
\\

\noindent
$\bullet$ {\em The kernel $\Theta_2$:}
Next, we consider the terms contracted by $B_{2,3}$ in \eqref{eq-ex-bound-1},
\eqn
    B_{2,3}U^{(3)}(t_{2}-t_{3}) \Big((|\phi\rangle\langle\phi|)^{\otimes 2}\otimes\Theta_3\Big)
    = \big(U^{(1)}(t_2-t_3)|\phi\rangle\langle\phi|\big)
    \otimes \Theta_2\,,
\eeqn
using \eqref{eq-B34-ex-1}, 
which defines the kernel  
\eqn\label{eq-ex-Theta2-def-1}
    \Theta_2(x,x')&=&
    B_{1,2}\Big(\big(U^{(1)}(t_2-t_3)|\phi\rangle\langle\phi|\big)\otimes(U^{(1)}(t_2-t_3)\Theta_3)\Big)(x,x')
    \nonumber\\
    &=&(U_{2,3}\phi)(x)\overline{(U_{2,3}\phi)(x')}
    \sum_{\beta_3=1}^2 c_{\beta_3}^3 \Big[(U_{2,3}\chi_{\beta_3}^3)(x)\overline{(U_{2,3}\psi_{\beta_3}^3)}(x)
    \nonumber\\
    &&\hspace{2cm}-(U_{2,3}\chi_{\beta_3}^3)(x')\overline{(U_{2,3}\psi_{\beta_3}^3)}(x')\Big]
    \nonumber\\
    &=:&\sum_{\beta_2=1}^4 c_{\beta_2}^2 \chi_{\beta_2}^2(x)\overline{\psi_{\beta_2}^2}(x') \,,
\eeqn
where
\eqn
    U_{i,j}:= e^{i(t_i-t_j)\Delta} \,.
\eeqn 
Since  for every $\beta_3$, only one out of the two factors $\chi_{\beta_3}^3$, $\psi_{\beta_3}^3$ is distinguished,  
it follows from \eqref{eq-ex-Theta2-def-1}
that for every $\beta_2\in\{1,\dots,4\}$, only one out of the two factors $\chi_{\beta_2}^2$, $\psi_{\beta_2}^2$ is distinguished.
The coefficients $c_{\beta_2}^2$ again have values in $\{1,-1\}$.
\\

\noindent
$\bullet$ {\em The kernel $\Theta_1$:}
Finally, we consider the terms contracted by $B_{1,2}$ in \eqref{eq-ex-bound-1}, corresponding to 
\eqn\label{eq-Theta1-ex-def-1}
    \Theta_1(x,x')&=&
    B_{1,2}\Big(\big(U^{(1)}(t_1-t_2)U^{(1)}(t_2-t_3)|\phi\rangle\langle\phi|\big)\otimes(U^{(1)}(t_1-t_2)\Theta_2)\Big)(x,x')
    \nonumber\\
    &=&(U_{1,3}\phi)(x)\overline{(U_{1,3}\phi)(x')}
    \sum_{\beta_2=1}^4 c_{\beta_2}^2 \Big[(U_{1,2}\chi_{\beta_2}^2)(x)\overline{(U_{1,2}\psi_{\beta_2}^2)}(x)
    \nonumber\\
    &&\hspace{2cm}-(U_{1,2}\chi_{\beta_2}^2)(x')\overline{(U_{1,2}\psi_{\beta_2}^2)}(x')\Big]
    \nonumber\\
    &=:&\sum_{\beta_1=1}^8 c_{\beta_1}^1 \chi_{\beta_1}^1(x)\overline{\psi_{\beta_1}^1}(x') \,.
\eeqn
Again, since  for every $\beta_2$, only one out of the two functions $\chi_{\beta_2}^2$, $\psi_{\beta_2}^2$ is distinguished,
it follows that for every $\beta_1\in\{1,\dots,8\}$, only one out of the two functions $\chi_{\beta_1}^1$, $\psi_{\beta_1}^1$ 
is distinguished. 
The coefficients $c_{\beta_1}^1$ again have values in $\{1,-1\}$. 
 
\subsubsection{Recursive bounds}

We may now return to \eqref{eq-ex-bound-1}, and perform the following recursive bounds with respect to time 
integration.

\noindent
$\bullet$ \underline{\em Integral in $t_1$.}
Applying Cauchy-Schwarz with respect to the integral in $t_1$ and the triangle inequality for the trace norm, we obtain that 
\eqn\label{eq-ex-bound-2}
    \eqref{eq-ex-bound-1}
    &=&\int_{[0,T)^3}dt_1 dt_2 dt_3
    \int d{\mu}_{t_3}^{(i)}(\phi) \tr\Big(\,\Big| \, U^{(1)}(t-t_1)\Theta_1\, \Big|\,\Big)
    \nonumber\\ 
    &\leq&\sum_{\beta_1=1}^8 T^{1/2}\int_{[0,T)^2}  dt_2 dt_3
    \int d{\mu}_{t_3}^{(i)}(\phi) 
     \Big\| \,  \|  \chi_{\beta_1}^1 \|_{L^2_x} \| \psi_{\beta_1}^1 \|_{L^2_x} \, \Big\|_{L^2_{t_1\in[0,T)}} \,, \;\;\;\;\;\;\;
\eeqn
using that $|c_{\beta_1}^1|=1$.
From \eqref{eq-Theta1-ex-def-1}, we see that given $\beta_1\in\{1,\dots,8\}$, there exists $\beta_2$ such that
\eqn
    \chi_{\beta_1}^1(x)&=&  (U_{1,3}\phi)(x)
    \nonumber\\
   \psi_{\beta_1}^1(x)& =&  (U_{1,3}\phi)(x) \overline{(U_{1,2}\chi_{\beta_2}^2)}(x) 
   (U_{1,2}\psi_{\beta_2}^2)(x) 
\eeqn
(or with a cubic expressions for $\chi_{\beta_1}^1$ and a linear expression for $ \psi_{\beta_1}^1$).
Therefore, 
\eqn\label{eq-ex-bound-3}
    \lefteqn{
    \Big\| \,  \|  \chi_{\beta_1}^1 \|_{L^2_x} \| \psi_{\beta_1}^1 \|_{L^2_x} \, \Big\|_{L^2_{t_1\in[0,T)}}
    }
    \nonumber\\ 
    &=&
    \|\phi\|_{L^2_{x}}\Big\| \, (U_{1,3}\phi)(x) \overline{ (U_{1,2}\chi_{\beta_2}^2)}(x)
    (U_{1,2}\psi_{\beta_2}^2)(x) \, \Big\|_{L^2_{t_1\in[0,T)}L^2_{x}}\,,
\eeqn
using that $U_{1,3}$ is unitary, and that $\phi$ does not depend on $t_1$.

Next, we observe that 
\eqn\label{eq-basicStrichartz-1}
    \lefteqn{
    \Big\|(e^{it\Delta}f_1)(x)\overline{(e^{it\Delta}f_2)(x)}(e^{it\Delta}f_3)(x) \Big\|_{L^2_{t}(\R)L^2_{x}(\R^3)}
    }
    \nonumber\\
    &\leq&
    \|e^{it\Delta}f_1\|_{L^\infty_t L^6_x}\|e^{it\Delta}f_2\|_{L^\infty_t L^6_x} 
    \|e^{it\Delta}f_3\|_{L^2_t L^6_x}
    \nonumber\\
    &\leq&
    C \|f_1\|_{H^1_x}\|f_2\|_{ H^1_x}   \|f_3\|_{L^2_x}
\eeqn    
using the H\"older inequality, the Sobolev inequality, and the Strichartz estimate
 $\|e^{it\Delta}f\|_{L^2_tL^6_x}\leq C \|f\|_{L^2}$  for
the free Schr\"odinger evolution. 
We make the important observation that in \eqref{eq-basicStrichartz-1},
we can place the $L^2_x$-norm 
on any of the three functions $f_j$, $j=1,2,3$, and not only on $f_3$.
Similarly,  if a derivative is included, 
\eqn\label{eq-basicStrichartz-1-2-1}
    \lefteqn{
    \Big\| \nabla_{x}( \, (e^{it\Delta}f_1)(x)  \overline{(e^{it\Delta}f_2)}(x)
   (e^{it\Delta}f_3)(x) \, ) \, \Big\|_{L^2_{t}(\R)L^2_{x}(\R^3)}
    }
    \nonumber\\
    &\leq& \sum_{j=1}^3
    \|e^{it\Delta}\nabla_x f_j\|_{L^2_t L^6_{x}}\prod_{1\leq i \leq 3\atop i\neq j}\|e^{it\Delta}f_i\|_{L^\infty_t L^6_{x}}  
    \nonumber\\
    &\leq&C\,
    \|f_1\|_{H^1_x}\|f_2\|_{H^1_{x}} \|f_3\|_{H^1_{x}} \,,
\eeqn  
which, together with  \eqref{eq-basicStrichartz-1}, implies that
\eqn\label{eq-basicStrichartz-1-2}
    \Big\|   (e^{it\Delta}f_1)(x)  \overline{(e^{it\Delta}f_2)}(x)
    (e^{it\Delta}f_3)(x) \, ) \, \Big\|_{L^2_{t}(\R)H^1_{x}(\R^3)}\leq C
    \prod_{j=1}^3\|f_j\|_{H^1_x}  \,.
\eeqn 
 
Only one of the factors $\chi_{\beta_2}^2$, $\psi_{\beta_2}^2$ is distinguished, say for instance $\psi_{\beta_2}^2$. 
We then use \eqref{eq-basicStrichartz-1} in such a way that the $L^2_x$-norm is applied to this term. All terms in
\eqref{eq-ex-bound-2} can be treated in the same manner, thus obtaining
\eqn\label{eq-ex-bound-1-2} 
    \eqref{eq-ex-bound-1}\leq C T^{1/2}
    \sum_{\beta_1=1}^8\int_{[0,T)^2}  dt_2 dt_3
    \int d{\mu}_{t_3}^{(i)}(\phi)  \|\phi\|_{H^1_x}^2\|\chi_{\beta_2}^2\|_{ H^1_x}  \|\psi_{\beta_2}^2\|_{L^2_x}\,,
\eeqn
where the indices $\beta_2$ depend on $\beta_1$.
Next, we use the defining relation \eqref{eq-ex-Theta2-def-1} for the functions $\chi_{\beta_2}^2$, $\psi_{\beta_2}^2$,
and consider the integral in $t_2$.

\noindent
$\bullet$ \underline{\em Integral in $t_2$.}
By assumption, the factor $\psi_{\beta_2}^2$ is distinguished, while $\chi_{\beta_2}^2$ is not.
Moreover, one of the functions $\chi_{\beta_2}^2$, $\psi_{\beta_2}^2$ is a linear, while the other one
is a cubic expression in the functions after the second equality sign in \eqref{eq-ex-Theta2-def-1}
(the distinguished factor could be either). Our goal is to bound the distinguished factor in $L^2$. 
From   comparing terms in  \eqref{eq-ex-Theta2-def-1}, one possible combination is
\eqn
  \chi_{\beta_2}^2(x)=(U_{2,3}\phi)(x) 
  \;\;\;,\;\;\;
  \psi_{\beta_2}^2(x) = 
  (U_{2,3}\phi)(x) \overline{(U_{2,3}\chi_{\beta_3}^3)}(x) (U_{2,3}\psi_{\beta_3}^3)(x)
  \,,
\eeqn
that is, the distinguished factor $\psi_{\beta_2}^2$ is a cubic expression.
We apply Cauchy-Schwarz in the $t_2$-integral in such a way that the $L^2_{t_2}$-norm falls on 
the cubic term.
(If, on the other hand, $\chi_{\beta_2}^2$ is the cubic term, we use Cauchy-Schwarz in $t_2$ to get
$\|\psi_{\beta_2}^2\|_{ L^\infty_{t_2\in[0,T)}L^2_x}$ and 
$\|\chi_{\beta_2}^2\|_{ L^2_{t_2\in[0,T)}H^1_x}\leq C \|\phi\|_{H^1_x}^3$ from  \eqref{eq-basicStrichartz-1-2}.)
We then get
\eqn 
    \eqref{eq-ex-bound-1-2} 
    &\leq& C T\sum_{\beta_1=1}^8
    \int_{[0,T)}  dt_3
    \int d{\mu}_{t_3}^{(i)}(\phi)  \|\phi\|_{H^1_x}^2
    \|\chi_{\beta_2}^2\|_{ L^\infty_{t_2\in[0,T)}H^1_x} \|\psi_{\beta_2}^2\|_{L^2_{t_2\in[0,T)}L^2_x}
    \nonumber\\
    &=&
    C T\sum_{\beta_1=1}^8
    \int_{[0,T)}  dt_3
    \int d{\mu}_{t_3}^{(i)}(\phi)  \|\phi\|_{H^1_x}^3
    \nonumber\\
    &&\hspace{1cm}
     \|(U_{2,3}\phi)(x)\overline{(U_{2,3}\chi_{\beta_3}^3)}(x) (U_{2,3}\psi_{\beta_3}^3)(x)\|_{L^2_{t_2\in[0,T)}L^2_x} \,,
\eeqn
where only one of the three factors inside the norm on the last line is distinguished.
We may assume it is $\psi_{\beta_3}^3$.
By comparing terms in \eqref{eq-ex-Theta3-def-1}, we then find that $\psi_{\beta_3}^3=\psidistex$, and  
$\chi_{\beta_3}^3=\phi$.
We then apply  \eqref{eq-basicStrichartz-1} again, and use the $L^2_x$-bound for $\psi_{\beta_3}^3=\psidistex$.
At this point, we substitute $\psidistex=|\phi|^2\phi$.

\noindent
$\bullet$ \underline{\em Using de Finetti for the last step.}
Subsequently, we obtain 
\eqn\label{eq-ex-laststep-1}
    \eqref{eq-ex-bound-1} &\leq&   C T \sum_{\beta_1=1}^8
    \int_{[0,T)}  dt_3
    \int d{\mu}_{t_3}^{(i)}(\phi)  \|\phi\|_{H^1}^5
     \| \psidistex\|_{ L^2_x} 
     \nonumber\\
     &\leq& 8 C T^2
    \sup_{t_3\in[0,T)}
    \int d{\mu}_{t_3}^{(i)}(\phi)  \|\phi\|_{H^1}^8
     \nonumber\\
     &\leq& 8 C T^2 M^4 \,,
\eeqn
where we used $\| \psidistex\|_{ L^2_x} \leq C\|\phi\|_{H^1}^3$ from Sobolev embedding, and  
the bound
\eqref{eq-muH1bound-1} related to the de Finetti theorem, which is uniform in $t_3$. 
This is the desired estimate in our example calculation.

The strategy presented in this example can be applied in the general case.

\centerline{\epsffile{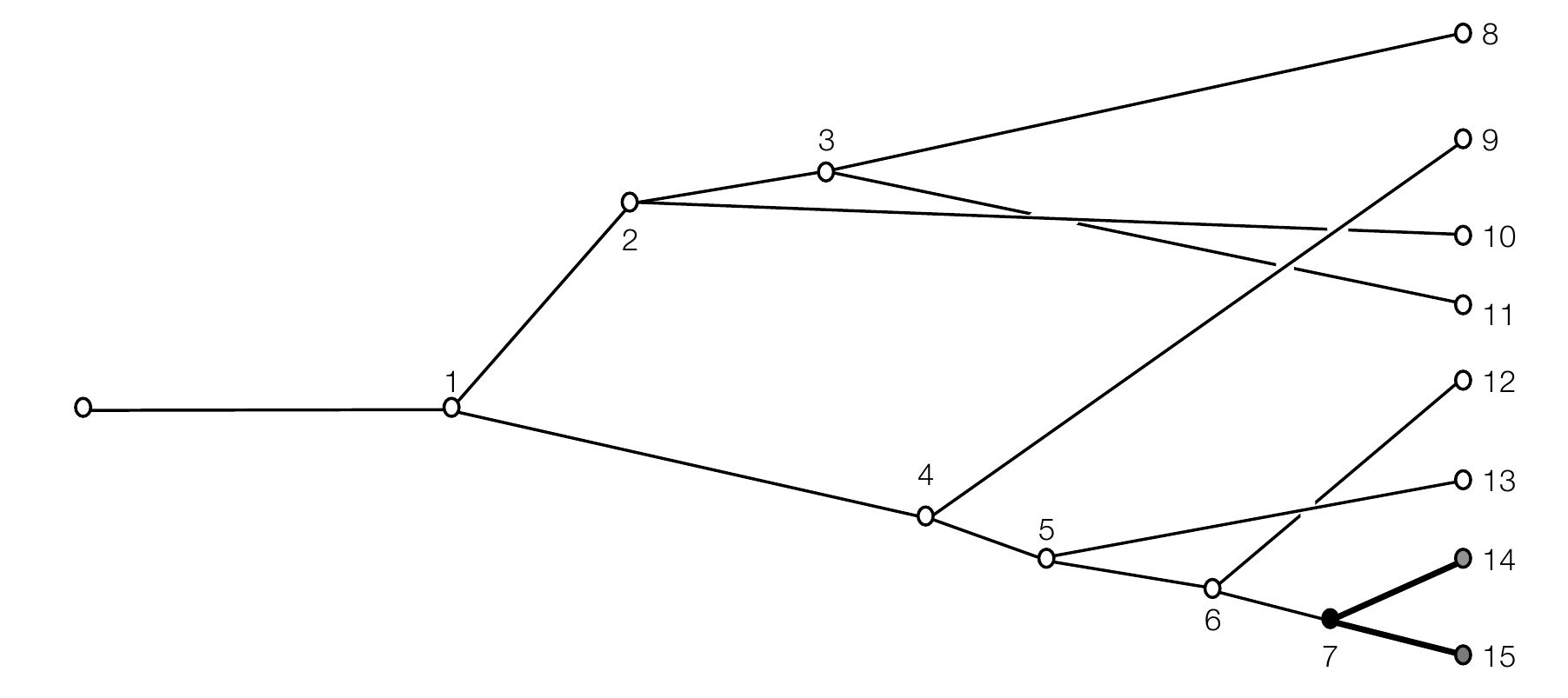} }

\noindent
Figure 2. An example of a distinguished tree $\tau_j$ with $m_j=7$. The number next to an internal vertex corresponds to its
label $\alpha\in\{1,\dots,7\}$, ordered increasingly from left to right. The vertex with label $\alpha$ corresponds to 
$B_{\sigma_j(\alpha+1),\alpha+1}$ (note the shift by 1 in the index). The time variable attached to it is $t_{\ell_{j,\alpha}}$.  
The leftmost vertex is the root vertex,
while the rightmost vertices are leaf vertices labeled with $\alpha\in\{8,15\}$,
corresponding to $u_{j,\alpha-7}$. The leaf vertices $8$ to $13$ are regular, while the 
leaf vertices $14,15$ are distinguished. 
The distinguished internal vertex has label $7$. Here, for example, $\kappa_-(1)=2$, $\kappa_+(1)=4$,
$\kappa_+(\kappa_+(\kappa_+(\kappa_+(1))))=7$, $\kappa_-(2)=3$.
This tree corresponds to the term 
\eqn
     U^{(1)}_{0,1}B_{1,2}U^{(2)}_{1,2}B_{1,3}U^{(3)}_{2,3}B_{1,4}U^{(4)}_{3,4}
     B_{2,5}U^{(5)}_{4,5}B_{5,6}U^{(6)}_{5,6}B_{5,7}U^{(7)}_{6,7}B_{7,8}(|\phi\rangle\langle\phi|)^{\otimes 8}
\eeqn
 (in upper echelon form) where $U^{(j)}_{i,j}=U^{(j)}(t_{i}-t_j)$ and $t_0=t$.

\subsection{Recursive definition of kernels at vertices}

As in the above example, we recursively assign a kernel $\Theta_\alpha$
to each vertex $\alpha$.
As the root of this induction, we associate the kernel 
\eqn\label{eq-Thetauij-def-1}
    \Theta_{\alpha}(x;x') := \phi(x)\overline{\phi(x')} 
\eeqn
to the leaf vertex with label $\alpha\in\{m_j+1,\dots,2m_j+2\}$ (corresponding to $u_{j,\alpha-m_j}$). 

In the first recursion step, we determine $\Theta_{m_j}$ at the distinguished vertex $\alpha=m_j$ from the term on
the last line of  \eqref{eq-Ijdisting-def-1}, given by
\eqn  
    B_{\sigma_j(m_j+1),m_j+1}\big( |\phi\rangle\langle \phi |\big)^{\otimes (\Lj)} 
    &=& \big( |\phi\rangle\langle \phi |\big)^{\otimes (\sigma(\Lj)-1)}
    \otimes \Theta_{m_j} 
    \nonumber\\
    &&\hspace{1cm}
    \otimes\big( |\phi\rangle\langle \phi |\big)^{\otimes (\Lj-\sigma(\Lj)-1)}
\eeqn
where 
\eqn\label{eq-Thetadist-def-1} 
    \Theta_{m_j}(x;x') := 
    \psidist(x)\overline{ \phi(x')} - \phi(x)\overline{\psidist(x')} 
\eeqn
with $\psidist$ as in \eqref{eq-psir-def-1}.  It is obtained from contracting the two
copies of $|\phi\rangle\langle\phi|$ at the two leaf vertices $\kappa_-(m_j)$, $\kappa_+(m_j)$
which have $m_j$ as their parent vertex. 

For the induction step, we
let $\alpha\in\{1,\dots,m_j-1\}$ label a regular internal vertex, and assume that the kernels $\Theta_{\alpha'}$
have been determined for all $\alpha'>\alpha$.  
Let $\kappa_-(\alpha)$, $\kappa_+(\alpha)$ label the two 
child vertices  (of internal or leaf type)  of $\alpha$,
\eqn 
    \sigma_j(\alpha)=\sigma_j(\kappa_-(\alpha)) 
    \;\;\;,\;\;\;
    \alpha=\sigma_j(\kappa_+(\alpha)) \,.
\eeqn
Then, by induction assumption, $\Theta_{\kappa_-(\alpha)}$, $\Theta_{\kappa_+(\alpha)}$ are given, and we define
\eqn\label{eq-Theta-def-1}
    \lefteqn{
    \Theta_{\alpha}(x;x') 
    }
    \nonumber\\
    &=& B_{1,2}\Big( \,
    \Big( \, U^{(1)}(t_{\alpha}-t_{\kappa_-(\alpha)}) \Theta_{\kappa_-(\alpha)} \, \Big)        
    \otimes \Big( \, U^{(1)}(t_{\alpha}-t_{\kappa_+(\alpha)})\Theta_{\kappa_+(\alpha)} \, \Big)\,\Big)(x;x')
    \nonumber\\
    &=&\Big( \, U^{(1)}(t_{\alpha}-t_{\kappa_-(\alpha)}) \Theta_{\kappa_-(\alpha)} \, \Big) (x;x')
    \Big[
    \Big( \, U^{(1)}(t_{\alpha}-t_{\kappa_+(\alpha)})\Theta_{\kappa_+(\alpha)} \, \Big)(x;x)
    \nonumber\\
    &&\hspace{3cm}
    -\Big( \, U^{(1)}(t_{\alpha}-t_{\kappa_+(\alpha)})\Theta_{\kappa_+(\alpha)} \, \Big)(x';x') 
    \Big] \,.
\eeqn 
Clearly, if $\kappa_\pm(\alpha)$ corresponds to a regular leaf vertex, then 
\eqn\label{eq-Thetaleaf-def-1}
    \Theta_{\kappa_\pm(\alpha)}(x;x') = \phi(x)\overline{\phi(x')}\,,
\eeqn
and $t_{\ell_{j,\kappa_\pm(\alpha)}}=t_{r}$. If  $\kappa_\pm(\alpha)=m_j$  is the
distinguished vertex,  we use  \eqref{eq-Thetadist-def-1}.

We iterate this procedure until we obtain the kernel $\Theta_1$ at $\alpha=1$, which is the unique child vertex of the
root vertex.

\subsection{Factorization structure of kernels}

We will now determine the structure of $\Theta_\alpha$.

\begin{lemma}
Let $\alpha\in\{1,\dots,m_j\}$. Then, every
kernel $\Theta_\alpha$ can  be written as a sum of differences of factorized kernels,
\eqn\label{eq-Theta-def-2}
    \Theta_\alpha(x;x') 
    &=&
    \sum_{\beta_\alpha }c_{\beta_\alpha }^\alpha  
    \chi_{\beta_\alpha}^\alpha(x) \overline{\psi_{\beta_\alpha}^\alpha(x')}
\eeqn
with at most $2^{m_j-\alpha}$ nonzero coefficients $c^\alpha_{\beta_\alpha}\in\{1,-1\}$.
\end{lemma}

\begin{proof} 
The kernels at the leaf vertices \eqref{eq-Thetauij-def-1} have the form \eqref{eq-Theta-def-2}.
If $\alpha=m_j$ is the distinguished vertex, $\Theta_{m_j}$ is given by \eqref{eq-Thetadist-def-1}, and evidently
has the form \eqref{eq-Theta-def-2}.
For the induction step, let us assume that given $\alpha\in\{1,\dots,m_j\}$,
the kernels $\Theta_{\alpha'}$ have the form \eqref{eq-Theta-def-2} for all $\alpha'>\alpha$, thus in particular
$\Theta_{\kappa_+(\alpha)}$, $\Theta_{\kappa_-(\alpha)}$ have this form. Then, from \eqref{eq-Theta-def-1}, we find that 
\eqn\label{eq-Thetaexpansion-1}
    \Theta_\alpha(x;x') 
    &=&\sum_{\beta_{\kappa_-(\alpha)},\beta_{\kappa_+(\alpha)}}
    c_{\beta_{\kappa_-(\alpha)}}^{\kappa_-(\alpha)}c_{\beta_{\kappa_+(\alpha)}}^{\kappa_+(\alpha)}
    (U_{\alpha;\kappa_-(\alpha)} \chi_{\beta_{\kappa_-(\alpha)}}^{\kappa_-(\alpha)})(x)
    \overline{(U_{\alpha;\kappa_-(\alpha)} \psi_{\beta_{\kappa_-(\alpha)}}^{\kappa_-(\alpha)})(x')}
    \nonumber\\
    &&\hspace{1cm}
    \Big[
    (U_{\alpha;\kappa_+(\alpha)} \chi_{\beta_{\kappa_+(\alpha)}}^{\kappa_+(\alpha)})(x)
    \overline{(U_{\alpha;\kappa_+(\alpha)} \psi_{\beta_{\kappa_+(\alpha)}}^{\kappa_+(\alpha)})(x)}
    \nonumber\\
    &&\hspace{2cm}
    -
    (U_{\alpha;\kappa_+(\alpha)} \chi_{\beta_{\kappa_+(\alpha)}}^{\kappa_+(\alpha)})(x')
    \overline{(U_{\alpha;\kappa_+(\alpha)} \psi_{\beta_{\kappa_+(\alpha)}}^{\kappa_+(\alpha)})(x')} \,
    \Big] \,,
\eeqn
where for brevity, we write 
\eqn 
    U_{\alpha;\alpha'} := 
    e^{i(t_{\ell_{j,\alpha}}-t_{\ell_{j,\alpha'}}) \Delta_x } 
     \;\;\;,\;\;\alpha,\alpha'\in\{0,1,\dots,m_j\}\,,
\eeqn 
(with $t_{0}:=t$) for free one-particle propagators.
We write the sum on the right hand side in an arbitrary but fixed order, and use the individual terms 
as definitions for the terms $c_{\beta_\alpha }^\alpha $,
$\chi_{\beta_\alpha}^\alpha$, $\psi_{\beta_\alpha}^\alpha$ in
\eqn\label{eq-Thetaalpha-prod-1}
    \Theta_\alpha(x;x') =
    \sum_{\beta_\alpha }c_{\beta_\alpha }^\alpha  
    \chi_{\beta_\alpha}^\alpha(x) \overline{\psi_{\beta_\alpha}^\alpha(x')} \,,
\eeqn
which is of the form \eqref{eq-Theta-def-2}.
This iteration terminates when we reach $\alpha=1$. 
\end{proof}

In particular, we have that
\eqn\label{eq-Ifirststep-1}
    \opJ_j^1(\sigma_j;t,t_{\ell_{j,1}},\dots,t_{\ell_{j,m_j}}) =
    U^{(1)}(t-t_{\ell_{j,1}})
    \Theta_1 \,.
\eeqn
For convenience, we will notationally suppress the dependence of the functions $ \chi_{\beta_\alpha}^\alpha$, 
$\psi_{\beta_\alpha}^\alpha$
on the time variables $t_{\ell_{\alpha'}}$, but we note that they do not depend on any  $t_{\ell_{\alpha'}}$ with $\alpha'<\alpha$.

\begin{definition}\label{def-distfact-1}
From here on, we re-interpret $\psidist$ in $\Theta_{m_j}$ (see   \eqref{eq-Thetadist-def-1}) as a function  
which is {\em independent} of $\phi$, $\overline{\phi}$.
Only at the end of our analysis, we will substitute $\psidist:=|\phi|^2\phi$. 
For each $\alpha=1,\dots,m_j$, we call a factor 
$\chi_{\beta_\alpha}^\alpha$,  $\psi_{\beta_\alpha}^\alpha$ in the expansion \eqref{eq-Theta-def-2} {\em distinguished} if it 
is a function of $\psidist$.
\end{definition}

Next, we derive recursive bounds on the functions $\chi_{\beta_\alpha}^\alpha$, $\psi_{\beta_\alpha}^\alpha$.

\subsection{Key properties of the kernels $\Theta_\alpha$}

We make the following key observations which will be crucial for the next steps of our proof:
\begin{itemize}
\item
The only dependence of $\Theta_\alpha$ on the time variable $t_{\ell_{j,\alpha}}$  is via the propagators 
\eqn 
    U_{\alpha;\kappa_\pm(\alpha)} = e^{i(t_{\ell_{j,\alpha}}-t_{\ell_{j,\kappa_\pm(\alpha)}}  )\Delta}
\eeqn
appearing on the right hand side of \eqref{eq-Thetaexpansion-1}. The kernels $\Theta_{\kappa_\pm(\alpha)}$ at 
the two child vertices $\kappa_\pm(\alpha)$ of $\alpha$ do not depend on  $t_{\ell_{j,\alpha}}$.
This will be crucial for the application of Strichartz estimates below.
\item
The product $\chi_{\beta_\alpha}^\alpha(x) \overline{\psi_{\beta_\alpha}^\alpha(x')} $ in \eqref{eq-Thetaalpha-prod-1}
either has the form
\begin{align}\label{eq-chipsi-prod-1}
    &\chi^\alpha_{\beta_{\alpha}}(x) \overline{\psi^\alpha_{\beta_{\alpha}} }(x')
    = (U_{\alpha;\kappa_-(\alpha)}\chi_{\beta_{\kappa_-(\alpha)}}^{\kappa_-(\alpha)})(x)
    \overline{ (U_{\alpha;\kappa_-(\alpha)}\psi_{\beta_{\kappa_-(\alpha)}}^{\kappa_-(\alpha)}) }(x')
    \nonumber\\
    &\hspace{3cm}
    (U_{\alpha;\kappa_+(\alpha)}\chi_{\beta_{\kappa_+(\alpha)}}^{\kappa_+(\alpha)})(x)
    \overline{ (U_{\alpha;\kappa_+(\alpha)} \psi_{\beta_{\kappa_+(\alpha)}}^{\kappa_+(\alpha)} ) }(x)
\end{align}
or
\begin{align}\label{eq-chipsi-prod-2}
    &\chi^\alpha_{\beta_{\alpha}}(x) \overline{\psi^\alpha_{\beta_{\alpha}} }(x')
    = (U_{\alpha;\kappa_-(\alpha)}\chi_{\beta_{\kappa_-(\alpha)}}^{\kappa_-(\alpha)})(x)
    \overline{ (U_{\alpha;\kappa_-(\alpha)}\psi_{\beta_{\kappa_-(\alpha)}}^{\kappa_-(\alpha)} ) }(x')
    \nonumber\\
    &\hspace{3cm}
    (U_{\alpha;\kappa_+(\alpha)}\chi_{\beta_{\kappa_+(\alpha)}}^{\kappa_+(\alpha)} )(x')
    \overline{ (U_{\alpha;\kappa_+(\alpha)}\psi_{\beta_{\kappa_+(\alpha)}}^{\kappa_+(\alpha)} ) }(x') \,,
\end{align}
for some values of $\beta_{\kappa_-(\alpha)}$, $\beta_{\kappa_+(\alpha)}$ that depend on $\beta_\alpha$.
Comparing the left and right hand sides, the function $\chi_{\beta_\alpha}^\alpha$ either  
has the cubic form
\eqn\label{eq-chi-kappa2-def-2}
    \lefteqn{
    \chi_{\beta_\alpha}^\alpha(x)  
    = 
    (U_{\alpha;\kappa_-(\alpha)} \chi_{\beta_{\kappa_-(\alpha)}}^{\kappa_-(\alpha)})(x) 
    }
    \nonumber\\
    &&\hspace{2cm}
    (U_{\alpha;\kappa_+(\alpha)} \chi_{\beta_{\kappa_+(\alpha)}}^{\kappa_+(\alpha)})(x)
    \overline{(U_{\alpha;\kappa_+(\alpha)} \psi_{\beta_{\kappa_+(\alpha)}}^{\kappa_+(\alpha)})(x)}\, ,
\eeqn
or the linear form
\eqn\label{eq-chi-kappa2-def-1}
    \chi_{\beta_\alpha}^\alpha(x)  
    = (U_{\alpha;\kappa_-(\alpha)} \chi_{\beta_{\kappa_-(\alpha)}}^{\kappa_-(\alpha)})(x) \,.
\eeqn
Accordingly, $\psi_{\beta_\alpha}^\alpha$  respectively has either linear or cubic form.
It is important that the product $\chi_{\beta_\alpha}^\alpha\overline{\psi_{\beta_\alpha}^\alpha}$
is always of quartic form \eqref{eq-chipsi-prod-1} or \eqref{eq-chipsi-prod-2}.
This fact will again be crucial for the application of Strichartz estimates below.
\item
In the product on the right hand side of \eqref{eq-chipsi-prod-1}, respectively \eqref{eq-chipsi-prod-2},
at most one of the four factors is distinguished (see Definition \eqref{def-distfact-1}).
This follows straightforwardly from  an induction along decreasing values of
$\alpha$, using the fact that the statement is true for all regular leaf vertices, and for the
distinguished vertex \eqref{eq-Thetadist-def-1}.
\end{itemize}

We may therefore make the following assumption, which leads to notational simplifications, but to
no loss of generality.

\begin{hypothesis}\label{hyp-L2assump-1}
In all that follows, we   assume for notational convenience that only the 
functions $\psi_{\beta_1}^{1}$, and recursively, 
$(\psi_{\beta_{\kappa_+^q(1) }}^{\kappa_+^q(1) })_{q=1}^{Q} $,
are distinguished (i.e., are a function of $\psidist$ in  \eqref{eq-Thetadist-def-1}).
Here,
\eqn
     \kappa_+^q(1):=\underbrace{\kappa_+(\kappa_+(\cdots(\kappa_+}_{q\;times}(1))\cdots)) \,,
\eeqn
is the $q$-th iterate of the index $\alpha=1$ under $\kappa_+$, and $Q$ is the number of edges
linking $\alpha=1$ to the distinguished vertex with label $\alpha=m_j$.
This is one special case, but all cases can be treated in the same way.
\end{hypothesis}

\section{Recursive  $L^2$- and $H^1$-bounds for the distinguished tree}
\label{sec-RecBounds-1}

From here on, we abbreviate the notation by writing $t_\alpha$ for $t_{\ell_{j,\alpha}}$,
and by referring to the vertex $v_{j,\alpha}$ by its label $\alpha$.
Also, we will say that the time variable $t_\alpha$ is attached to the vertex $\alpha$.

We have that
\begin{align}\label{eq-IjInt-def-1}
     &\int_{[0,T)^{m_j-1}}dt_1\cdots dt_{m_j-1}
     \tr\Big(\,\Big| \, \opJ_j^1(\sigma_j;t,t_{1},\dots,t_{m_j}) \, \Big| \, \Big) 
      \nonumber\\
    &\hspace{1cm}=
    \int_{[0,T)^{m_j-1}}dt_1\cdots dt_{m_j-1} 
     \tr\big( \, \big| \, U^{(1)}(t-t_1) \Theta_1 \, \big| \, \big) 
     \nonumber\\ 
    &\hspace{1cm}\leq \sum_{\beta_{1}} 
    \int_{[0,T)^{m_j-1}}dt_1\cdots dt_{m_j-1} 
     \|\psi_{\beta_{1}}^1\|_{L^2}
    \|\chi_{\beta_{1}}^1\|_{L^2}\,.
\end{align}
We will estimate the last term on the right hand side
based on the recursion formula \eqref{eq-Thetaexpansion-1}, using recursive bounds
adapted to a hierarchy of subtrees of $\tau_j$.

Our main goal is to propagate the $L^2$-norm in \eqref{eq-IjInt-def-1} along edges of $\tau_j$ that connect the 
vertex $\alpha=1$ 
to the distinguished vertex $\alpha=m_j$, 
in order to obtain a bound 
\eqn\label{eq-psimj-bound-1}
    \|\psidist\|_{L^2}=\|\,|\phi|^2\phi\|_{L^2}\leq C \|\phi\|_{H^1}^3
\eeqn 
which can be controlled
by the growth condition \eqref{eq-muH1bound-1}.

\subsection{Recursive bounds}

We let $\tau_{j,\alpha}$ denote the subtree of $\tau_j$ with root at 
the vertex labeled by $\alpha$. 
Moreover, we denote by
\eqn
    \int\Big[\prod_{\alpha'\in \tau_{j,\alpha}} dt_{\alpha'}\Big] \equiv
    \int_{[0,T)^{d_\alpha}}\Big[\prod_{\alpha'\in \tau_{j,\alpha}} dt_{\alpha'}\Big]
\eeqn 
integration with respect
to all time variables attached to the internal and root vertices of the subtree $\tau_{j,\alpha}$ with root at $\alpha$.
Here, $d_\alpha$ denotes the total number of internal and root vertices of $\tau_{j,\alpha}$.

\begin{lemma}\label{lm-iterationH1L2-bd-1}
Let $\kappa_-(\alpha)$ and $\kappa_+(\alpha)$ label the two child vertices of 
the vertex labeled by $\alpha$. Assume that either \eqref{eq-chipsi-prod-1} or \eqref{eq-chipsi-prod-2} is given.
Then, the following recursive bounds hold:
\begin{itemize}

\item
Bound on $L^2$-level
\begin{align}\label{eq-L2levelbd-1}
    &\int \Big[\prod_{\alpha'\in \tau_{j,\alpha}} dt_{\alpha'}\Big]\|\psi_{\beta_{\alpha}}^\alpha\|_{L^2}
    \|\chi_{\beta_{\alpha}}^\alpha\|_{H^1}
    \leq C T^{\frac12}
    \int \Big[\prod_{\alpha'\in \tau_{j,\kappa_-(\alpha)}} dt_{\alpha'}\Big] 
    \|\psi_{\beta_{\kappa_-(\alpha)}}^{\kappa_-(\alpha)}\|_{H^1}
    \|\chi_{\beta_{\kappa_-(\alpha)}}^{\kappa_-(\alpha)}\|_{H^1} 
    \nonumber\\
    &\hspace{2cm}\cdot
    \int \Big[\prod_{\alpha'\in \tau_{j,\kappa_+(\alpha)}} dt_{\alpha'}\Big]
    \|\psi_{\beta_{\kappa_+(\alpha)}}^{\kappa_+(\alpha)}\|_{L^2}
    \|\chi_{\beta_{\kappa_+(\alpha)}}^{\kappa_+(\alpha)}\|_{H^1} 
\end{align}
 
\item
Bound on $H^1$-level
\begin{align}\label{eq-H1levelbd-1}
    &\int \Big[\prod_{\alpha'\in \tau_{j,\alpha}} dt_{\alpha'}\Big]\|\psi_{\beta_{\alpha}}^\alpha\|_{H^1}
    \|\chi_{\beta_{\alpha}}^\alpha\|_{H^1}
    \leq  C T^{\frac12}
    \int \Big[\prod_{\alpha'\in \tau_{j,\kappa_-(\alpha)}} dt_{\alpha'}\Big] 
    \|\psi_{\beta_{\kappa_-(\alpha)}}^{\kappa_-(\alpha)}\|_{H^1}
    \|\chi_{\beta_{\kappa_-(\alpha)}}^{\kappa_-(\alpha)}\|_{H^1} 
    \nonumber\\
    &\hspace{2cm}\cdot
    \int \Big[\prod_{\alpha'\in \tau_{j,\kappa_+(\alpha)}} dt_{\alpha'}\Big]
    \|\psi_{\beta_{\kappa_+(\alpha)}}^{\kappa_+(\alpha)}\|_{H^1}
    \|\chi_{\beta_{\kappa_+(\alpha)}}^{\kappa_+(\alpha)}\|_{H^1}
\end{align}

\end{itemize}
\end{lemma}

\begin{proof}
This can be inferred from the bounds \eqref{eq-basicStrichartz-1} and \eqref{eq-basicStrichartz-1-2}
as follows:

\noindent
$\bullet$ \underline{\em Bound on $L^2$-level.}
By applying \eqref{eq-basicStrichartz-1} to 
\eqref{eq-chi-kappa2-def-2} and \eqref{eq-chi-kappa2-def-1} with respect to the time variable $t_\alpha$, and recalling that
\eqn
    U_{\alpha;\kappa_-(\alpha)} = e^{i(t_\alpha-t_{\kappa_-(\alpha)})\Delta} \,,
\eeqn
we obtain 
\eqn\label{eq-intIbound-L2-2}
    \lefteqn{
     \int_{[0,T)^{d_\alpha}} \Big[\prod_{\alpha'\in \tau_{j,\alpha}} dt_{\alpha'}\Big] 
     \|\psi_{\beta_{\alpha}}^\alpha\|_{L^2}
    \|\chi_{\beta_{\alpha}}^\alpha\|_{H^1}
    }
    \nonumber\\ 
    &\leq& 
    C T^{1/2}\int_{[0,T)^{d_\alpha-1}} \Big[\prod_{\alpha'\in \tau_{j,\kappa_-(\alpha)}\cup \tau_{j,\kappa_+(\alpha)}} dt_{\alpha'}\Big]
    \Big\| \,  e^{-it_{\kappa_-(\alpha)}\Delta}\chi_{\beta_{\kappa_-(\alpha)}}^{\kappa_-(\alpha)} \, \Big\|_{H^1}
     \nonumber\\
    &&\hspace{1cm}
    \Big\| \, e^{-it_{\kappa_-(\alpha)}\Delta}\psi_{\beta_{\kappa_-(\alpha)}}^{\kappa_-(\alpha)} \, \Big\|_{H^1}
    \Big\| \, e^{-it_{\kappa_+(\alpha)}\Delta}\chi_{\beta_{\kappa_+(\alpha)}}^{\kappa_+(\alpha)} \, \Big\|_{H^1}
    \Big\| \, e^{-it_{\kappa_+(\alpha)}\Delta}\psi_{\beta_{\kappa_+(\alpha)}}^{\kappa_+(\alpha)} \, \Big\|_{L^2}
    \nonumber\\ 
    &=& 
    C T^{1/2} 
    \int_{[0,T)^{d_{\kappa_-(\alpha)}}} \Big[\prod_{\alpha'\in \tau_{j,\kappa_-(\alpha)} } dt_{\alpha'}\Big]
    \Big\| \,  \chi_{\beta_{\kappa_-(\alpha)}}^{\kappa_-(\alpha)} \, \Big\|_{H^1}
    \Big\| \, \psi_{\beta_{\kappa_-(\alpha)}}^{\kappa_-(\alpha)} \, \Big\|_{H^1}
     \nonumber\\
    &&\hspace{1cm}
    \int_{[0,T)^{d_{\kappa_+(\alpha)}} } \Big[\prod_{\alpha'\in  \tau_{j,\kappa_+(\alpha)}} dt_{\alpha'}\Big]
    \Big\| \, \chi_{\beta_{\kappa_+(\alpha)}}^{\kappa_+(\alpha)} \, \Big\|_{H^1}
    \Big\| \, \psi_{\beta_{\kappa_+(\alpha)}}^{\kappa_+(\alpha)} \, \Big\|_{L^2} \,.
\eeqn
Here, we first used Cauchy-Schwarz in the $t_\alpha$-integral.
In the last step, we used that $\psi_{\beta_{\alpha}}^\alpha$, $\chi_{\beta_{\alpha}}^\alpha$
depend only on the time variables $t_{\alpha'}$ attached to the vertices of the subtree $\tau_{j,\alpha}$
rooted at the vertex $\alpha$, for every $\alpha\in\{1,\dots,m_j-1\}$.
Moreover, we used that $e^{-it_{\kappa_\pm(\alpha)}\Delta}$ are unitary in $L^2$ and $H^1$. 

\noindent
$\bullet$ \underline{\em Bound on $H^1$-level.}
Using \eqref{eq-basicStrichartz-1-2}, we obtain
\eqn\label{eq-intIbound--H12}
    \lefteqn{
     \int_{[0,T)^{d_\alpha}} \Big[\prod_{\alpha'\in \tau_{j,\alpha}} dt_{\alpha'}\Big] 
     \|\psi_{\beta_{\alpha}}^\alpha\|_{H^1}
    \|\chi_{\beta_{\alpha}}^\alpha\|_{H^1}
    }
    \nonumber\\ 
    &\leq& 
    C T^{1/2}\int_{[0,T)^{d_\alpha-1}} \Big[\prod_{\alpha'\in \tau_{j,\kappa_-(\alpha)}\cup 
    \tau_{j,\kappa_+(\alpha)}} dt_{\alpha'}\Big]
    \Big\| \,  e^{-it_{\kappa_-(\alpha)}\Delta}\chi_{\beta_{\kappa_-(\alpha)}}^{\kappa_-(\alpha)} \, \Big\|_{H^1}
     \nonumber\\
    &&\hspace{1cm}
    \Big\| \, e^{-it_{\kappa_-(\alpha)}\Delta}\psi_{\beta_{\kappa_-(\alpha)}}^{\kappa_-(\alpha)} \, \Big\|_{H^1}
    \Big\| \, e^{-it_{\kappa_+(\alpha)}\Delta}\chi_{\beta_{\kappa_+(\alpha)}}^{\kappa_+(\alpha)} \, \Big\|_{H^1}
    \Big\| \, e^{-it_{\kappa_+(\alpha)}\Delta}\psi_{\beta_{\kappa_+(\alpha)}}^{\kappa_+(\alpha)} \, \Big\|_{H^1}
    \nonumber\\ 
    &=& 
    C T^{1/2} 
    \int_{[0,T)^{d_{\kappa_-(\alpha)}}} \Big[\prod_{\alpha'\in \tau_{j,\kappa_-(\alpha)} } dt_{\alpha'}\Big]
    \Big\| \,  \chi_{\beta_{\kappa_-(\alpha)}}^{\kappa_-(\alpha)} \, \Big\|_{H^1}
    \Big\| \, \psi_{\beta_{\kappa_-(\alpha)}}^{\kappa_-(\alpha)} \, \Big\|_{H^1}
     \nonumber\\
    &&\hspace{1cm}
    \int_{[0,T)^{d_{\kappa_+(\alpha)}} } \Big[\prod_{\alpha'\in  \tau_{j,\kappa_+(\alpha)}} dt_{\alpha'}\Big]
    \Big\| \, \chi_{\beta_{\kappa_+(\alpha)}}^{\kappa_+(\alpha)} \, \Big\|_{H^1}
    \Big\| \, \psi_{\beta_{\kappa_+(\alpha)}}^{\kappa_+(\alpha)} \, \Big\|_{H^1} \,,
\eeqn
by proceeding as above for the bounds on the $L^2$-level. 
\end{proof}

\section{Concluding the proof}
\label{sec-prooffinish-1}

Using Lemma \ref{lm-iterationH1L2-bd-1}, we can now prove the following
main estimates for the distinguished tree in Proposition \ref{prop-mainIjest-1},
and for regular trees in Proposition \ref{prop-regularIjest-1}.

\begin{proposition}\label{prop-mainIjest-1}
Assume that $\tau_j$ is the distinguished tree. Then, the bound
\eqn\label{eq-intIbound-3} 
     \lefteqn{
     \int_{[0,T)^{m_j-1}}dt_1\cdots dt_{m_j-1} 
     \tr\Big( \, \Big| \, \opJ_j^1(\sigma_j;t,t_{1},\dots,t_{m_j}) \, \Big| \, \Big) 
    }
    \nonumber\\
    &&\hspace{2cm}
    \leq \; 
    2^{m_j} C^{m_j}  T^{(m_j-1)/2} \|\phi\|_{H_1}^{\Lj} \,
\eeqn 
holds. 
\end{proposition}

\begin{proof}
To begin with,
\eqn
    \lefteqn{
     \int_{[0,T)^{m_j-1}}dt_1\cdots dt_{m_j-1} 
     \tr\Big(\,\Big| \, \opJ_j^1(\sigma_j;t,t_{\ell_1},\dots,t_{m_j}) \, \Big|\,\Big)  
    }
    \nonumber\\
    &=&
    \int_{[0,T)^{m_j-1}}dt_1\cdots dt_{m_j-1} 
     \tr\big( \, \big| \,  U^{(1)}(t-t_1) \Theta_1 \, \big| \, \big) 
     \nonumber\\ 
    &\leq& \sum_{\beta_{1}} 
    \int_{[0,T)^{m_j-1}}dt_1\cdots dt_{m_j-1} 
     \|\psi_{\beta_{1}}^1\|_{L^2}
    \|\chi_{\beta_{1}}^1\|_{L^2}
     \nonumber\\
    &\leq& \sum_{\beta_{\kappa_-(1)},\beta_{\kappa_+(1)}} 
    C T^{1/2} 
    \int \Big[\prod_{\alpha'\in \tau_{j,\kappa_-(1)}} dt_{\alpha'}\Big] 
    \|\psi_{\beta_{\kappa_-(1)}}^{\kappa_-(1)}\|_{H^1}
    \|\chi_{\beta_{\kappa_-(1)}}^{\kappa_-(1)}\|_{H^1} 
    \label{eq-intIbound-1}\\
    &&\hspace{1cm}\cdot
    \int \Big[\prod_{\alpha'\in \tau_{j,\kappa_+(1)}} dt_{\alpha'}\Big]
    \|\psi_{\beta_{\kappa_+(1)}}^{\kappa_+(1)}\|_{L^2}
    \|\chi_{\beta_{\kappa_+(1)}}^{\kappa_+(1)}\|_{H^1}  \,, 
    \label{eq-intIbound-2}
\eeqn 
where we first used \eqref{eq-Ifirststep-1}, then \eqref{eq-Thetaalpha-prod-1},
and subsequently \eqref{eq-L2levelbd-1} with respect to the integral in $t_1$, at the vertex $\alpha=1$.
Then, we used \eqref{eq-Thetaexpansion-1}, and the fact that 
$|c_{\beta_{\kappa_i(1)}}^{\kappa_i(1)}|=1$.
By Hypothesis \ref{hyp-L2assump-1}, $\psi_{\beta_{\kappa_+(1)}}^{\kappa_+(1)}$ is the only function on the
last two lines that is distinguished, which is why it is the only term bounded in $L^2$.

We first bound the integral \eqref{eq-intIbound-1}. To this end, we iterate the bound \eqref{eq-H1levelbd-1} on
the $H^1$-level until we reach all leaf vertices of the subtree $\tau_{j,\kappa_-(1)}$.
It follows from Hypothesis \ref{hyp-L2assump-1} that $\tau_{j,\kappa_-(1)}$ does not contain the distinguished
vertex, therefore all leaf vertices of $\tau_{j,\kappa_-(1)}$ are regular.
Then,  we find that
\eqn\label{eq-recurs-bounds-1}
    \int \Big[\prod_{\alpha'\in \tau_{j,\kappa_-(1)}} dt_{\alpha'}\Big] 
    \|\psi_{\beta_{\kappa_-(1)}}^{\kappa_-(1)}\|_{H^1}
    \|\chi_{\beta_{\kappa_-(1)}}^{\kappa_-(1)}\|_{H^1} \leq  
    C T^{d_{\kappa_-(1)}/2}\|\phi\|_{H^1}^{2b_{\kappa_-(1)}}
\eeqn
where $b_{\alpha}$ is the number of regular leaf vertices of the subtree  $\tau_{j,\alpha}$ rooted at $\alpha$,
and $d_{\alpha}$ is the number of internal vertices  of the subtree  $\tau_{j,\alpha}$.

Next, we bound the integral \eqref{eq-intIbound-2}. To this end, we iterate both the bound \eqref{eq-L2levelbd-1} on
the $L^2$-level, and the bound  \eqref{eq-H1levelbd-1} on
the $H^1$-level until we reach all leaf vertices of the subtree $\tau_{j,\kappa_+(1)}$, including the distinguished vertex
$v_{j,m_j}$.  The $L^2$ norm is in every step applied to the
term $\psi_{\beta_{\kappa_+(\alpha)}}^{\kappa_+(\alpha)}$, which is the only function
in \eqref{eq-chipsi-prod-1}, respectively \eqref{eq-chipsi-prod-2}, which is distinguished,
due to Hypothesis  \ref{hyp-L2assump-1}. 
The iteration terminates when all regular leaf vertices, and the distinguished vertex are reached.
We then obtain that
\eqn\label{eq-recurs-bounds-2}
    \lefteqn{
    \int \Big[\prod_{\alpha'\in \tau_{j,\kappa_+(1)}} dt_{\alpha'}\Big]
    \|\psi_{\beta_{\kappa_+(1)}}^{\kappa_+(1)}\|_{L^2}
    \|\chi_{\beta_{\kappa_+(1)}}^{\kappa_+(1)}\|_{H^1}
    }
    \nonumber\\
    &&\leq 
     C^{m_j} T^{(d_{\kappa_+(1)}-1)/2)}\|\phi\|_{H^1}^{2b_{\kappa_+(1)}} \|\psidist\|_{L^2} 
    \nonumber\\
    &&\leq
    C^{m_j}  T^{(d_{\kappa_+(1)}-1)/2)}\|\phi\|_{H^1}^{2b_{\kappa_+(1)}+3} \,
\eeqn
where the $L^2$-norm has been moved to the distinguished vertex, hence the factor $\|\psidist\|_{L^2}$.
At this point, we substituted $\psidist:=|\phi|^2\phi$, and 
used the Sobolev embedding.

Combining the bounds \eqref{eq-recurs-bounds-1} and \eqref{eq-recurs-bounds-2}, we obtain 
\eqref{eq-intIbound-3}, where all leaves contribute a factor $\|\phi\|_{H^1}^2$.
The factor $2^{m_j}$ bounds the number of terms in the sum over 
$\beta_{\kappa_-(1)}$, $\beta_{\kappa_+(1)}$ in \eqref{eq-intIbound-1}.
\end{proof}

Similarly, we find for regular trees:

\begin{proposition}\label{prop-regularIjest-1}
Assume that $\tau_j$ is a regular tree. Then, the bound
\eqn\label{eq-intIbound-4} 
     \lefteqn{
     \int_{[0,T)^{m_j}}dt_1\cdots dt_{m_j} 
     \tr\Big( \, \Big| \, \opJ_j^1(\sigma_j;t,t_{1},\dots,t_{m_j}) \, \Big| \, \Big)  
    }
    \nonumber\\
    &&\hspace{2cm}
    \leq \; 
    2^{m_j} C^{m_j}  T^{m_j/2} \|\phi\|_{H_1}^{\Lj} \,
\eeqn 
holds. 
\end{proposition}

\begin{proof}
For a regular tree, we have 
\begin{align}\label{eq-gammaj-reduced-1} 
    &\opJ_j^1(\sigma_j;t,t_{1},\dots,t_{m_j})=U^{(1)}(t-t_{1})
    B_{\sigma_j(2),2}\cdots 
    \\
    &\hspace{2cm}
    \cdots
    U^{(m_j)}(t_{m_j-1}-t_{m_j} ) B_{\sigma_j(m_j+1),m_j+1}U^{(m_j+1)}(t_{m_j})
    \big( |\phi\rangle\langle \phi |\big)^{\otimes (\Lj) } \,.
    \nonumber
\end{align}
The key difference between \eqref{eq-gammaj-reduced-1}  and the corresponding expression 
\eqref{eq-Ijdisting-def-1} for the distinguished tree is the presence of the free propagator 
$U^{(m_j+1)}(t_{m_j})$ on the last line.
The proof is immediately obtained from the proof of Proposition \ref{prop-mainIjest-1}, by
using 
\begin{align}\label{eq-intIbound-4}  
     &\int_{[0,T)^{m_j}}dt_1\cdots dt_{m_j} 
     \tr\Big( \, \Big| \, \opJ_j^1(\sigma_j;t,t_{1},\dots,t_{m_j}) \, \Big| \, \Big)
     \nonumber\\  
    &\hspace{1cm}=
    \int_{[0,T)^{m_j}}dt_1\cdots dt_{m_j} 
     \tr\big( \, \big| \, U^{(1)}(t-t_1) \Theta_1 \, \big| \, \big) 
     \nonumber\\ 
    &\hspace{1cm}\leq  \sum_{\beta_{1}} 
    \int_{[0,T)^{m_j}}dt_1\cdots dt_{m_j} 
     \|\psi_{\beta_{1}}^1\|_{L^2}
    \|\chi_{\beta_{1}}^1\|_{L^2}
     \nonumber\\ 
    &\hspace{1cm}\leq  \sum_{\beta_{1}} 
    \int_{[0,T)^{m_j}}dt_1\cdots dt_{m_j} 
     \|\psi_{\beta_{1}}^1\|_{H^1}
    \|\chi_{\beta_{1}}^1\|_{H^1} \,,
\end{align}
and by iterating the bound \eqref{eq-H1levelbd-1} on
the $H^1$-level until we reach all leaf vertices of $\tau_j$. Because all leaves of $\tau_j$ are
regular, no $L^2$-level bound is necessary.
\end{proof}

Going back to \eqref{eq-gammak-L2bound-1}, we find 
from
\eqn
    \opJ^k(\sigma;t,t_1,\dots,t_r;\ux_k;\ux_k') = 
    \prod_{j=1}^k \opJ_j^1(\sigma_j;t,t_{\ell_{j,1}},\dots,t_{\ell_{j,m_j}};x_j;x_j')
\eeqn
that
\eqn\label{eq-gammak-L2bound-3}
    \lefteqn{
    \int_{[0,t]^{r-1}}dt_1\cdots dt_{r-1}  \tr\Big( \,\Big| \, \opJ^k(\sigma;t,t_1,\dots,t_r)\, 
    \Big| \, \Big) 
    }
    \nonumber\\
    &=&
    \int_{[0,t]^{r-1}}dt_1\cdots dt_{r-1}   \prod_{j=1}^k   \tr\Big( \, \Big| \, \opJ_j^1
    (\sigma_j;t,t_{\ell_{j,1}},\dots,t_{\ell_{j,m_j}}) \, \Big| \, \Big)  
    \nonumber\\
    &\leq&   2^{r}  C^{r}  T^{(r-1)/2}   \|\phi\|_{H_1}^{2(k+r)} \,,
\eeqn
by combining the estimates from the $k-1$ regular trees, and from the distinguished tree.
The factor $2^r$ is obtained from the product of factors $2^{m_j}$ from all trees $\tau_j$,
both regular and distinguished.

Then, we observe that for $t\in[0,T)$,
\eqn 
    \tr\big( \, \big|\gamma^{(k)}(t)\, \big| \,\big)     
    &\leq&
    (2 C T)^{(r-1)/2} 
    \sum_{i=1,2}\int_0^T dt_{r} \int d{\mu}_{t_r}^{(i)}(\phi)
    \|\phi\|_{H^1}^{2(r+k)} 
    \nonumber\\
    &\leq& 
    (2 C T)^{(r+1)/2} \, \sup_{t_r\in[0,T)}\sum_{i=1,2}\int d\mu_{t_r}^{(i)}(\phi)
    \|\phi\|_{H^1}^{2(r+k)}   \,.
    \label{eq-Bgamma-L2L2-bd-1}
\eeqn
The growth condition \eqref{eq-muH1bound-1} implies that this is bounded by
\eqn
    \eqref{eq-Bgamma-L2L2-bd-1}
    &<& 2 \, M^{2k-2} \, (2 C M^4 T)^{(r+1)/2}  \, \longrightarrow \, 0 \;\;\; (r\rightarrow \infty)
\eeqn
for $T<(2 C M^4)^{-1}$   (which is in particular uniform in $k$). Since $k$ is fixed and $r$ is arbitrary,
we conclude that
\eqn 
   \tr\big( \, \big|\gamma^{(k)}(t)\, \big| \,\big)  = 0  \,, \;\;\; t\in[0,T) \,,
\eeqn
which implies that $\gamma^{(k)}(t)=0$ for $t\in[0,T)$, and hence, uniqueness holds.

Moreover, it can be easily checked that 
\eqn 
    \gamma^{(k)}(t) = \int d\mu(\phi)(|S_t(\phi)\rangle\langle S_t(\phi)|)^{\otimes k} 
    \;\;\;,\;\;\;\forall k\in\N\,,
\eeqn
is a mild solution of the GP hierarchy in $L^\infty_{t\in[0,T)}\frH^1$ with initial data
\eqn 
     \gamma^{(k)}(0) =  \int d\mu(\phi)(|\phi\rangle\langle\phi|)^{\otimes k} 
    \;\;\;,\;\;\;\forall k\in\N\,.
\eeqn     
By uniqueness, it is the only such solution.
This proves  Theorem \ref{thm-main-1}.
\qed

\subsection*{Acknowledgements} 
We thank Benjamin Schlein, Mathieu Lewin, and an anonymous referee for very helpful remarks.
The work of T.C. was supported by NSF grants DMS-1009448
and DMS-1151414 (CAREER). 
The work of N.P. was supported by NSF grant DMS-1101192.
The work of R.S. was supported by NSERC.

\end{document}